\documentclass[11pt,a4paper,reqno]{amsart}%
\usepackage{amsthm,amsmath,amsfonts,amssymb,amsxtra,appendix,bookmark,dsfont,latexsym}
\usepackage{mathrsfs}
\usepackage{tikz}
\usetikzlibrary{arrows}
\usepackage{verbatim}

\makeatletter
\usepackage{hyperref}
\usepackage{delarray,a4,color}
\usepackage{euscript}
\usepackage[latin1]{inputenc}

\definecolor{darkred}{cmyk}{0,1,1,0.4}

\makeatletter
\def\@tocline#1#2#3#4#5#6#7{\relax
  \ifnum #1>\c@tocdepth % then omit
  \else
    \par \addpenalty\@secpenalty\addvspace{#2}%
    \begingroup \hyphenpenalty\@M
    \@ifempty{#4}{%
      \@tempdima\csname r@tocindent\number#1\endcsname\relax
    }{%
      \@tempdima#4\relax
    }%
    \parindent\z@ \leftskip#3\relax \advance\leftskip\@tempdima\relax
    \rightskip\@pnumwidth plus4em \parfillskip-\@pnumwidth
    #5\leavevmode\hskip-\@tempdima
      \ifcase #1
       \or\or \hskip 1em \or \hskip 2em \else \hskip 3em \fi%
      #6\nobreak\relax
      \dotfill
      \hbox to\@pnumwidth{\@tocpagenum{#7}}
    \par
    \nobreak
    \endgroup
  \fi}
\makeatother

\theoremstyle{plain}
\newtheorem{theorem}{Theorem}[section]
\newtheorem{lemma}[theorem]{Lemma}
\newtheorem{corollary}[theorem]{Corollary}
\newtheorem{proposition}[theorem]{Proposition}

\theoremstyle{remark}
\newtheorem{remark}[theorem]{Remark}

\numberwithin{equation}{section}

\DeclareMathOperator{\supp}{supp}

\def\bq{\begin{eqnarray}}
\def\eq{\end{eqnarray}}
\def\bqq{\begin{eqnarray*}}
\def\eqq{\end{eqnarray*}}
\def\nn{\nonumber}

\def\eps{\varepsilon}

\newcommand{\norm}[1]{\left\lVert #1 \right\rVert}

\newcommand\1{{\ensuremath {\mathds 1} }}

\renewcommand{\epsilon}{\varepsilon}

\def\cF {\mathcal{F}}

\def\R {\mathbb{R}}
\def\NN {\mathbb{N}}

\def\cE {\mathcal{E}}
\def\cQ {\mathcal{Q}}

\def\R {\mathbb{R}}

\renewcommand{\leq}{\leqslant}
\renewcommand{\geq}{\geqslant}

\newcommand{\bA}{\mathbf{A}}
\newcommand{\bJ}{\mathbf{J}}

\newcommand{\nablap}{\nabla^{\perp}}
\newcommand{\EAF}{E ^{\mathrm{af}}}
\newcommand{\cEAF}{\cE ^{\mathrm{af}}}
\newcommand{\uAF}{u ^{\mathrm{af}}}
\newcommand{\rhoAF}{\varrho^{\mathrm{af}}}
\newcommand{\rhoBAF}{\bar{\varrho}^{\mathrm{af}}}
\DeclareMathOperator{\curl}{\mathrm{curl}}

\newcommand{\cDAF}{\mathscr{D}^{\mathrm{af}}}

\newcommand{\cEGP}{\cE ^{\mathrm{GP}}}

\newcommand{\utest}{u^{\rm test}}

\newcommand{\cETF}{\mathcal{E} ^{\rm TF}}
\newcommand{\ETF}{E ^{\rm TF}}
\newcommand{\rhoTF}{\varrho^{\rm TF}}
\newcommand{\lTF}{\lambda^{\rm TF}}

\DeclareMathOperator{\dotconv}{\raisebox{-3pt}{\textup{\textbf{*}}}}
\numberwithin{equation}{section}
\newcommand{\bdm}{\begin{displaymath}}
\newcommand{\edm}{\end{displaymath}}
\newcommand{\bdn}{\begin{eqnarray}}
\newcommand{\edn}{\end{eqnarray}}
\newcommand{\bay}{\begin{array}{c}}
\newcommand{\eay}{\end{array}}
\newcommand{\ben}{\begin{enumerate}}
\newcommand{\een}{\end{enumerate}}
\newcommand{\beq}{\begin{equation}}
\newcommand{\eeq}{\end{equation}}
\newcommand{\beqn}{\begin{eqnarray}}
\newcommand{\eeqn}{\end{eqnarray}}
\newcommand{\bml}[1]{\begin{multline} #1 \end{multline}}

\def\Xint#1{\mathchoice
{\XXint\displaystyle\textstyle{#1}}%
{\XXint\textstyle\scriptstyle{#1}}%
{\XXint\scriptstyle\scriptscriptstyle{#1}}%
{\XXint\scriptscriptstyle\scriptscriptstyle{#1}}%
\!\int}
\def\XXint#1#2#3{{\setbox0=\hbox{$#1{#2#3}{\int}$ }
\vcenter{\hbox{$#2#3$ }}\kern-.6\wd0}}

\def\dashint{\Xint-}

\newcommand{\lf}{\left}
\newcommand{\ri}{\right}

\newcommand{\nv}{\mathbf{n}}
\newcommand{\xv}{\mathbf{x}}
\newcommand{\yv}{\mathbf{y}}

\newcommand{\zv}{\mathbf{z}}

\newcommand{\diff}{\mathrm{d}}

\newcommand{\dist}{\mathrm{dist}}

\newcommand{\tx}{\textstyle}

\renewcommand{\leq}{\leqslant}
\renewcommand{\geq}{\geqslant}

\renewcommand{\le}{\leqslant}
\renewcommand{\ge}{\geqslant}

\title[On the almost-bosonic anyon gas]{Local density approximation for the almost-bosonic anyon gas} 

\author[M. Correggi]{Michele CORREGGI}
\address{Dipartimento di Matematica ``G. Castelnuovo'', Universit\`{a} degli Studi di Roma ``La Sapienza'', P.le Aldo Moro, 5, 00185, Rome, Italy.}
\email{michele.correggi@gmail.com}

\author[D. Lundholm]{Douglas LUNDHOLM}
\address{KTH Royal Institute of Technology, Department of Mathematics, SE-100 44 Stockholm, Sweden}
\email{dogge@math.kth.se}

\author[N. Rougerie]{Nicolas ROUGERIE}
\address{CNRS \& Universit\'e Grenoble Alpes, LPMMC (UMR 5493), B.P. 166, F-38042 Grenoble, France}
\email{nicolas.rougerie@lpmmc.cnrs.fr}

\date{April 2017}

\begin{document}

\begin{abstract}
We study the minimizers of an energy functional with a self-consistent magnetic field,
which describes a quantum gas of almost-bosonic anyons in the average-field approximation. 
For the homogeneous gas we prove the existence of the thermodynamic limit of 
the energy at fixed effective statistics parameter, and the independence of such a limit from the shape of the domain. 
This result is then used in a local density approximation to derive an 
effective Thomas--Fermi-like model for the trapped anyon gas in the limit of a  large effective statistics parameter (i.e., ``less-bosonic'' anyons).
\end{abstract}

\maketitle

\setcounter{tocdepth}{2}
\tableofcontents

\section{Introduction}\label{sec:intro}

A convenient description of 2D particles with exotic quantum statistics (different from Bose--Einstein and Fermi--Dirac) is via effective magnetic interactions. We are interested in a mean-field model for such particles, known as anyons. Indeed, in a certain scaling limit (``almost-bosonic anyons'', see~\cite{LunRou-15}), a suitable magnetic non-linear Schr\"odinger theory becomes appropriate. The corresponding energy functional is given by 
\begin{equation}\label{eq:avg func}
%	\cEAF\colon H^1(\R^2) \to \R^+, \qquad
	\cEAF_\beta[u] := \int_{\R ^2} \left( \left| \left( -i\nabla +  \beta \bA [|u|^2] \right) u \right|^2 + V|u|^2 \right), %\diff \xv 
%	\qquad u \in H^1(\R^2),
\end{equation}
acting on functions $u \in H^1(\R^2)$.
Here $V\colon \R ^2 \to \R^+$ 
is a trapping potential confining the particles, and the vector potential 
$\bA [|u|^2]\colon \R^2 \to \R^2$ is defined through 
\begin{equation}\label{eq:avg field}
	\bA [\varrho] := \nablap w_0 \ast \varrho, 
	\qquad w_0 (\xv) := \log |\xv|,
\end{equation}
for $\varrho = |u|^2 \in L^1(\R^2)$ and $\xv^\perp = (x,y)^\perp := (-y,x)$.
Thus, the self-consistent magnetic field, given by
$$
	\curl \bA[\varrho] (\xv) = \Delta w_0 * \varrho(\xv) = 2\pi \varrho(\xv),
$$
is proportional to the particles' density. The parameter $\beta \in \R$ then regulates the strength of the magnetic 
self-interactions and, for reasons explained below, we will call it the \emph{scaled statistics parameter}.
Note that by symmetry of \eqref{eq:avg func} under complex conjugation 
$u \mapsto \overline{u}$ we may and shall assume 
$$ %\beq
	\beta \ge 0
$$ %\eeq
in the following. We will study the ground-state problem for~\eqref{eq:avg func}, namely the minimization under the mass constraint
\begin{equation}\label{eq:mass}
 \int_{\R^2} |u| ^2 = 1. 
\end{equation}

The functional $\cEAF$ bears some similarity with other mean-field models such as the Gross--Pitaevskii energy functional
\begin{equation}\label{eq:GP func}
	%\cEGP\colon H^1(\R^2) \to \R^+, \qquad
	\cEGP [u] := \int_{\R ^2} \left( |{-i}\nabla u + \mathbf{A} u |^2 + V|u|^2 + g |u| ^4  \right), %\diff \xv 
\end{equation}
with \emph{fixed} vector potential $\mathbf{A}$. 
The above describes a gas of interacting bosons in a certain mean-field regime~\cite{LieSeiSolYng-05,LieSei-06,NamRouSei-15,Rougerie-LMU,Rougerie-spartacus}: the quartic term originates from short-range pair interactions. The crucial difference between~\eqref{eq:avg func} and~\eqref{eq:GP func} is that, while the
interactions of $\cEGP$ are scalar (with interaction strength $g \in \R$), those of $\cEAF$ are purely magnetic and therefore involve mainly the phase of the function $u$. There is an extensive literature dealing with~\eqref{eq:GP func} (see~\cite{Aftalion-07,CorPinRouYng-12,CorRouYng-11,CorRou-13} for references) and with the related Ginzburg--Landau model of superconductivity~\cite{BetBreHel-94,FouHel-10,SanSer-07,Sigal-13}. That the interactions are via the magnetic field in~\eqref{eq:avg func} poses however quite a few new difficulties in the asymptotic analysis of minimizers we initiate here. Note indeed (see the variational equation in Lemma~\ref{lem:var eq}) that the non-linearity consists in a quintic non-local semi-linear term and a cubic quasi-linear term (also non-local), both being critical when compared to the usual Laplacian.

The functional $\cEAF$ arises in a mean-field description\footnote{Usually refered to as an 
\emph{average}-field description in this context.} 
of a gas of particles whose many-body quantum wave function can change under particle exchange by a phase factor $e^{i\alpha\pi}$ 
(with $\alpha \in \R$ known as the statistics parameter). This is a generalization of the usual types of particles: bosons have $\alpha=0$ (symmetric wave functions) and their mean-field description is via models of the form~\eqref{eq:GP func}, fermions have $\alpha=1$ (anti-symmetric wave functions) and appropriate models for them are Hartree--Fock functionals (see~\cite{Bach-92,LieSim-77,Lions-87,Lions-88,FouLewSol-15} and references therein). 
For general $\alpha$ one speaks of anyons~\cite{Khare-05,Myrheim-99,Ouvry-07,Wilczek-90}, which are believed to emerge as quasi-particle excitations of certain condensed-matter systems~\cite{AroSchWil-84,Haldane-83,Halperin-84,ZhaSreGemJai-14,CooSim-15,LunRou-16}. 

Anyons can be modeled as bosons (respectively, fermions) but with a many-body magnetic interaction of coupling strength $\alpha$ (respectively, $\alpha-1$). It was shown in \cite{LunRou-15} that the  ground-state energy per particle of such a system is correctly described 
by the minimum of \eqref{eq:avg func} (and the ground states by the corresponding minimizers) in a limit where, as the number of particles $N \to \infty$, 
one takes $\alpha = \beta/N \to 0$. We refer to this limit as that of \emph{almost-bosonic} anyons, with $\beta$ determining how far we are from usual bosons. 

\medskip

In the following we treat the anyon gas as fully described by a one-body wave function $u \in H^1(\R^2)$ minimizing~\eqref{eq:avg func} under the mass constraint~\eqref{eq:mass}. We shall consider asymptotic regimes for this minimization problem. The limit $\beta \to 0$ is 
trivial and leads to a linear theory for non-interacting bosons (see~\cite[Appendix~A]{LunRou-15}). The limit $\beta \to \infty$ is more interesting and 
more physically relevant: In a physical situation the statistics parameter $\alpha$ is fixed and finite and $N$ large, so that taking $\beta \to \infty$ is the relevant regime, at least if one is allowed to exchange the two limits.

In an approximation that has been used frequently in the physics literature~\cite{ChiSen-92,IenLec-92,LiBhaMur-92,Trugenberger-92b,Trugenberger-92,WenZee-90,Westerberg-93},
the ground-state energy per particle of the $N$-particle anyon gas 
with statistics parameter $\alpha$ is given by
\begin{equation} \label{eq:avg field approx}
	\frac{E_0(N)}{N} \approx \int_{\R^2} \left( 2\pi|\alpha| N\varrho^2 + V\varrho \right).
\end{equation}
This relies on assuming that each particle sees the others by their approximately constant \emph{average} magnetic field $B(\xv) \approx 2\pi\alpha N\varrho(\xv)$, with $\varrho(\xv)$ the local particle density  (normalized to $\int_{\R^2} \varrho = 1$). In the ground state of this magnetic field (the lowest Landau level) this leads to a magnetic energy $|B| \approx 2\pi|\alpha| N\varrho$ per particle\footnote{Because of the periodicity of the exchange phase 
$e^{i\alpha\pi}$, it is known that such an approximation can only be valid for certain
values of $\alpha$ and $\varrho$. See~\cite{LarLun-16,Lundholm-16,Trugenberger-92b} for further discussion.}. 

In this work we prove that, for large $\beta$, the behavior of the functional~\eqref{eq:avg func} is captured at leading order by a Thomas--Fermi type~\cite{CatBriLio-98,Lieb-81b} energy functional of a form similar to the right-hand side of \eqref{eq:avg field approx} with $|\alpha|N = \beta$. The coupling constant appearing in this functional is defined via the large-volume limit of the homogeneous anyon gas energy (i.e., the infimum of~\eqref{eq:avg func} confined to a bounded domain with $V=0$). 
In particular we prove that this limit exists and is bounded from below by 
the value $2\pi$ predicted by~\eqref{eq:avg field approx}. 
We do not know the exact value, but there are good reasons to believe that it is {\it not} equal to $2\pi$, thus refining the simple approximations leading to~\eqref{eq:avg field approx}.  

\medskip 

We state our main theorems in Section~\ref{sec:results} and present their proofs in Sections~\ref{sec:thermo} and~\ref{sec:trap}. Appendix~\ref{sec:var eq} recalls a few facts concerning the minimizers of~\eqref{eq:avg func}. In particular, although we do not need it for the proof of our main results, we derive the associated variational equation.

\medskip

\noindent\textbf{Acknowledgments:} This work is supported by MIUR through the FIR grant 2013 ``Condensed Matter in Mathematical Physics (Cond-Math)'' (code RBFR13WAET), the Swedish Research Council (grant no. 2013-4734) and the ANR (Project Mathostaq ANR-13-JS01-0005-01). 
We thank Jan Philip Solovej for insightful suggestions and Romain Duboscq for inspiring numerical simulations.
D.L. also thanks Simon Larson for discussions.

\section{Main results}\label{sec:results}

We now proceed to state our main theorems. We first discuss the large-volume limit for the homogeneous gas in Subsection~\ref{sec:intro thermo} 
and then state our results about the trapped anyons functional~\eqref{eq:avg func} in Subsection~\ref{sec:intro trap}.

\subsection{Thermodynamic limit for the homogeneous gas}\label{sec:intro thermo}

Let $\Omega \subset \R^2$ be a fixed bounded domain in $\R^2$, with the 
associated energy for almost-bosonic anyons confined to it:
\begin{equation}\label{eq:func dom}
	\cEAF_\Omega[u] = \cEAF_{\Omega,\beta}[u] 
	:= \int_{\Omega} \left| \left( -i\nabla + \beta \bA [|u|^2] \right) u \right|^2,
\end{equation}
with
\begin{equation}\label{eq:A on domain}
	\bA[|u|^2](\xv) = \int_\Omega \nablap w_0(\xv-\yv)|u(\yv)|^2 \,\diff\yv.
\end{equation}
We define two energies; with homogeneous Dirichlet boundary conditions 
\begin{equation}\label{eq:ener Dir}
	E_0 (\Omega, \beta, M) := \inf\left\lbrace \cEAF_{\Omega,\beta} [u] : u \in H^1_0(\Omega), \: \int_{\Omega} |u| ^2 = M \right\rbrace,
\end{equation}
and without boundary conditions,
\begin{equation}\label{eq:ener Neu}
	E (\Omega, \beta, M) := \inf\left\lbrace \cEAF_{\Omega,\beta} [u] : u \in H^1(\Omega), \: \int_{\Omega} |u| ^2 = M\right\rbrace.
\end{equation}
Of course, the last minimization leads to a magnetic Neumann boundary condition 
for the solutions. We are interested in the thermodynamic limit of these quantities, 
i.e., the scaling limit in which the size of the domain tends to $\infty$ with fixed density $\rho := M/|\Omega|$
and the normalization changes accordingly. 

\begin{theorem}[\textbf{Thermodynamic limit for the homogeneous anyon gas}]\label{thm:thermo limit}\mbox{}\\
	Let $\Omega \subset \R ^2$ be a bounded simply connected domain with 
	Lipschitz boundary, $\beta \geq 0$ and $\rho \geq 0$ be fixed parameters. Then, the limits 
	\begin{equation}\label{eq:thermo limit}
	 e(\beta , \rho) :=  \lim_{L\to \infty} \frac{E (L \Omega, \beta, \rho L ^2|\Omega|)}{L ^2|\Omega|} = \lim_{L\to \infty} \frac{E_0 (L \Omega, \beta, \rho L ^2|\Omega|)}{L ^2|\Omega|}
	\end{equation}
	exist, coincide and are independent of $\Omega$. Moreover
	\begin{equation}\label{eq:scale thermo}
	 e(\beta,\rho) = \beta \rho ^2 e (1,1).
	\end{equation}
\end{theorem}

\begin{remark}[Error estimate]
	\mbox{}	\\
	A close inspection of the proof reveals that we also have an estimate of the error appearing in \eqref{eq:thermo limit}, which coincides with the error appearing in the estimate of the difference between the Neumann and Dirichlet energies in a box (Lemma \ref{lem:Neu Dir}). Such a quantity is expected to be of the order of the box's side length $ L  $, which is subleading if compared to the total energy of order $ L^2 $. Our error estimate $ O(L^{12/7+\eps}) $ (see \eqref{eq:error estimate}) is however much larger and far from being optimal.  
	\hfill$\diamond$
\end{remark}

The above result defines the thermodynamic energy per unit area at scaled statistics parameter $\beta$ and density $\rho$, denoted $e(\beta,\rho)$,  and shows that it has a nice  scaling property. The latter is responsible for the occurrence of a Thomas--Fermi-type functional in the trapped anyons case.  The fact that $e(\beta,\rho)$ does not depend on boundary conditions is a crucial technical ingredient in our study of the trapped case.
This is very different from the usual Schr\"odinger energy in a fixed external
magnetic field, for example a constant one, for which the type of boundary
conditions do matter (see e.g. \cite[Chapter~5]{FouHel-10}).

The constant $e(1,1)$ will be used to define a corresponding coupling parameter below. One may observe that 
(see Lemma~\ref{lem:thermo Dir})
\begin{equation}\label{eq:const bound}
	e(1,1) \geq 2\pi,
\end{equation}
and we conjecture that this inequality is actually \emph{strict},
contrary to what might be expected when comparing to the coupling constant of the
conventional (constant-field)
average-field approximation \eqref{eq:avg field approx}.
The reason for this is that the self-interaction encoded by the functional
$\cEAF$ has not been fully incorporated in \eqref{eq:avg field approx}.
In fact, the lower bound \eqref{eq:const bound} is based on a magnetic $L^4$-bound
(Lemma~\ref{lem:mag ineq}) which is saturated only for constant functions,
and hence for constant densities, which certainly is compatible with 
\eqref{eq:avg field approx} in the case of homogeneous traps.
On the other hand, in order to minimize the magnetic energy in \eqref{eq:func dom}
for large $\beta$, the function has to have a large phase circulation 
and therefore also a large vorticity. This suggests the formation of an approximately homogeneous vortex 
lattice, in some analogy to the Abrikosov lattice that arises in superconductivity 
and in rotating bosonic gases~\cite{Aftalion-07,CorYng-08,SanSer-07}.
Such a picture has already been hinted at in \cite[p.~1012]{ChenWilWitHal-89}
for the almost-bosonic gas.
However the implication that the actual coupling constant may then be larger than the 
one expected from \eqref{eq:avg field approx} seems not to have been observed 
in the literature before.

One should note here  
that there is a certain abuse of language in using the term 
``thermodynamic limit''. Indeed, we consider the large-volume behavior of a 
mean-field energy functional, and there is no guarantee that this rigorously 
approximates the true thermodynamic energy of the underlying many-body system. 
%In fact there may even be issues in defining the latter [\dred{ref}].
%\end{remark}

\subsection{Local density approximation for the trapped gas}\label{sec:intro trap}

We now return to~\eqref{eq:avg func} and discuss the ground state problem 
\begin{equation}\label{eq:GSE}
	\EAF_\beta  := \min \left\{ \cEAF_\beta [u] : u \in H^1(\R^2), \ V |u|^2 \in L^1(\R^2), \int_{\R ^2}  |u| ^2 = 1 \right\}. 
\end{equation}
We denote $\uAF$ any associated minimizer. We refer to the Appendix 
(see also \cite[Appendix~A]{LunRou-15}) for a discussion of the 
minimization domain as well as the %proof of the 
existence of a minimizer. 
In the limit $\beta \to \infty$, the following simpler Thomas--Fermi (TF) like 
functional emerges 
\begin{equation}\label{eq:TF func}
	\cETF [\varrho] = \cETF_\beta[\varrho]
	:= \int_{\R ^2} \left( \beta e(1,1) \varrho ^2 + V \varrho \right), 
\end{equation}
whose ground-state energy we denote
\begin{equation}\label{eq:TF ener}
	\ETF_\beta := \min \left\{ \cETF_\beta [\varrho] : \varrho \in L^2(\R^2; \R^+), \: V \varrho \in L^1(\R^2), \: \int_{\R ^2}  \varrho = 1 \right\},
\end{equation}
with associated (unique) minimizer $\rhoTF_{\beta}$. 
Here $e(1,1)$ is the fixed, universal constant defined by Theorem~\ref{thm:thermo limit}. 

A typical potential one could have in mind for physical relevance 
is a harmonic trap, $V(\xv) = c|\xv|^2$,
or an asymmetric trap, $V(x,y) = c_1 x^2 + c_2 y^2$.
We shall work under the assumption that $V$ is homogeneous of degree $s$ and smooth:
\begin{equation}\label{eq:asum V}
	V (\lambda \xv) = \lambda ^s V (\xv), \qquad V\in C^{\infty} (\R^2).
\end{equation}
These conditions can be relaxed significantly; %a lot, 
in particular we could %straightforwardly 
extend the approach to asymptotically 
homogeneous potentials as defined in~\cite[Definition 1.1]{LieSeiYng-01}. 
We refrain from doing so to avoid lengthy technical discussions in the proofs.
We shall always %also 
impose that $V$ is trapping in the sense that it grows super-linearly at infinity, 
i.e., $s > 1$ and
\begin{equation}\label{eq:growth V}
	\min_{|\xv|\geq R} V (\xv) \xrightarrow[R\to \infty]{} \infty.
\end{equation}

The TF problem \eqref{eq:TF ener}
has the merit of being exactly soluble. We obtain by scaling
\begin{equation}\label{eq:TF scaling}
\ETF_\beta = \beta ^{s/(s+2)} \ETF_1, \qquad \rhoTF_\beta(\xv) = \beta^{-2/(2+s)} \rhoTF_1 \left( \beta ^{-1/(s+2)} \xv \right), 
\end{equation}
and by an explicit computation 
\begin{equation}\label{eq:TF exp}
 \rhoTF_1 (\xv) = \frac{1}{2 e(1,1)} \left( \lTF_1 - V(\xv)\right)_+,
\end{equation}
with the chemical potential
\begin{equation}\label{eq:TF chem}
 \lTF_1 = \ETF_1 + e(1,1) \int_{\R^2} \left(\rhoTF_1 \right) ^2.
\end{equation}
%For simplicity we shall also make the non-degeneracy assumption 
%\begin{equation}\label{eq:non degen}
% |\partial_{\nv} V| \neq 0,	\quad	\mbox{a.e. on } \partial \supp (\rhoTF_1),
%\end{equation}
%where $\nv$ denotes the (say outward) normal vector to $\partial \supp (\rhoTF_1)$.
%This is trivially satisfied if $V$ is radially symmetric.

Clearly the above considerations imply 
\begin{equation}\label{eq:TF support}
 \supp (\rhoTF_\beta) \subset B_{C \beta ^{1/(2+s)}} (0),
\end{equation}
where $B_{R}(\xv)$ stands for a ball (disk) of radius $R$ centered at $\xv$, 
and the estimates
\begin{align}\label{eq:TF estim}
\norm{\rhoTF_\beta}_{L^{\infty} (\R^2)} &\leq C \beta ^{-2/(2+s)}, \nonumber\\
\norm{\nabla \rhoTF_\beta}_{L^{\infty} (\R^2)} &\leq C \beta ^{-3/(2+s)},
\end{align}
for some fixed constant $C>0$. Noticing that $\rhoTF_1$ vanishes along a
level curve of the smooth homogeneous potential $V$,
we also have the non-degeneracy
\begin{equation}\label{eq:non degen}
 |\partial_{\nv} V| \neq 0,	\quad	\mbox{a.e. on } \partial \supp (\rhoTF_1),
\end{equation}
where $\nv$ denotes the (say outward) normal vector to $\partial \supp (\rhoTF_1)$.
% \medskip

We have the following result showing the accuracy of TF theory to determine the leading order of the minimization problem~\eqref{eq:GSE}: 

\begin{theorem}[\textbf{Local density approximation for the anyon gas}]\label{thm:main}\mbox{}\\
	Let $V$ satisfy %\eqref{eq:asum V}, \eqref{eq:growth V} and \eqref{eq:non degen}.
	\eqref{eq:asum V} and \eqref{eq:growth V}.
	In the limit $\beta \to \infty$ we have the energy convergence
	\begin{equation}\label{eq:ener CV}
	 \lim_{\beta \to + \infty} \frac{\EAF_\beta}{\ETF_{\beta}} = 1.   
	\end{equation}
	Moreover, for any function $\uAF$ achieving the infimum~\eqref{eq:GSE}, with 
	$\rhoAF := |\uAF| ^2$, we have for any $R>0$
	\begin{equation}\label{eq:dens CV}
	 \norm{ \beta ^{2/(2+s)} \rhoAF \left( \beta^{1/(2+s)} \, . \, \right) - \rhoTF_1}_{W^{-1,1} (B_R(0))} \xrightarrow[\beta \to \infty]{} 0  
	\end{equation}
	where $W^{-1,1} (B_R(0))$ is the dual space of Lipschitz functions on the ball $B_R(0)$.
\end{theorem}

\begin{remark}[Extension to more general potentials] \label{rem:gen pot}
	\mbox{}	\\
	The result can be straightforwardly %easily 
	extended to asymptotically homogeneous potentials, i.e., functions 
	$V(\xv)$ that satisfy the following 
	property \cite[Definition 1.1]{LieSeiYng-01}: there exists another function 
	$\tilde{V}$, non-vanishing for $\xv \neq 0$, such that, for some $s > 0$,
	\beq
		\lim_{\lambda \to \infty} \frac{\lambda^{-s} V(\lambda \xv) - \tilde{V}(\xv)}{1 + |\tilde{V}(\xv)|} = 0,
	\eeq
	uniformly in $ \xv \in \R^2 $. The function $ \tilde{V} $ is necessarily %easily proven to be 
	homogeneous of degree $ s > 0 $ and, if we denote by $ \tilde{\mathcal{E}}_{\beta} $ 
	the TF functional \eqref{eq:TF func} with $\tilde{V}$ in place of $V$, we have
	\bdm
		E^{\mathrm{TF}}_{\beta} = (1 + o(1)) \tilde{E}^{\mathrm{TF}}_{\beta},
		\qquad
		\tilde{E}^{\mathrm{TF}}_{\beta} = \beta^{\frac{s}{s+2}} \tilde{E}^{\mathrm{TF}}_{1},
		\qquad \text{as} \ \beta \to \infty.
	\edm
	\hfill$\diamond$
\end{remark}
 
\begin{remark}[Density approximation on finer length scales] \label{rem:density}
	\mbox{}	\\
	We conjecture that the estimate~\eqref{eq:dens CV} can be improved to show that  $ \rhoAF $ is close to $\rhoTF_\beta$ on finer scales. Namely~\eqref{eq:dens CV} implies that they are close on length scales of order $\beta ^{1/(2+s)}$ which is the extent of the support of $\rhoTF_\beta$, but we expect them to be close on scales $\gg \beta ^{-s/(2(s+2))}$, which is the smallest length scale appearing in our proofs. We however believe that the density convergence \emph{cannot} hold on scales smaller than $\beta ^{-s/(2(s+2))}$, for we expect the latter to be the length scale of a vortex lattice developed by minimizers.	
	\hfill$~\diamond$
\end{remark}

\begin{remark}[Large $\beta$ limit for the homogeneous gas on bounded domains] 
	\label{rem:td limit}
	\mbox{}	\\
	We can think of the homogeneous gas by 
	formally taking the limit $ s \to \infty $ of the homogeneous potentials we 
	have considered so far, which naturally leads to the restriction of 
	the functional $\cEAF$ in \eqref{eq:avg func} to bounded domains $\Omega$ 
	with $V = 0$ and Dirichlet boundary conditions, that is 
	\eqref{eq:func dom}-\eqref{eq:ener Dir}. 
	In fact, we have by the %same method of proof as for Theorem~\ref{thm:main},
	scaling laws discussed in Section~\ref{sec:scale},
	\beq
		\label{eq:td limit bd domains rem}
		\lim_{\beta \to +\infty} \frac{E(\Omega, \beta, 1)}{\beta} = \lim_{\beta \to +\infty} \frac{E_0(\Omega, \beta, 1)}{\beta} = |\Omega|^{-1} e(1,1),
	\eeq
	for any bounded and simply connected $\Omega$ with Lipschitz boundary. 
	Convergence of the 
	density to the TF minimizer $\rhoTF_1 $ holds true in the same form as in 
	\eqref{eq:dens CV}. In this case $\rhoTF_1$ is simply the constant function
	on the domain (confirming that the gas is indeed homogeneous).
	%Note that for bounded domains $\rhoTF_1$ is simply constant
	%in order to minimize the quartic norm under $L^1$ normalization.
	The shortest length scale on which we expect (but cannot prove) the density convergence is $\beta^{-1/2}$, which should be the typical length scale of the vortex structure. \hfill$\diamond$ 
\end{remark}

\section{Proofs for the homogeneous gas}\label{sec:thermo}

The basic ingredient of the proof for the inhomogeneous case is the 
understanding of the thermodynamic limit of the model where the trap is 
replaced by a finite domain with sharp walls. We discuss this here, proving 
Theorem~\ref{thm:thermo limit} and defining the constant $e(1,1)$ appearing in 
the TF functional~\eqref{eq:TF func}.
For the sake of concreteness we first set
\beq \label{eq:thermo def}
	e(\beta , \rho) := \liminf_{L\to \infty} \frac{E_0 (L \Omega, \beta, \rho L ^2|\Omega|)}{L ^2|\Omega|}
\eeq
for $ \Omega $ equal to a unit square and observe that such a quantity certainly 
exists %since the r.h.s. is positive. 
and is non-negative.
At this stage it might as well be infinite 
but we are going to prove that actually the limit exists, it is finite, and 
furthermore independent of the domain shape. 

We briefly anticipate here the proofs plan: next Section \ref{sec:tools thermo} 
contains basic technical estimates that we are going to use throughout the 
paper; Section \ref{sec:scale} contains the proof of a crucial scaling property 
of the energy in the homogeneous case; in Section \ref{sec:proof thermo} we 
prove the existence of the thermodynamic limit for the case of squares, and then extend the result to general domains.

\subsection{Toolbox}\label{sec:tools thermo}

Let us gather a few lemmas that will be used repeatedly in the sequel. 
We start with a variational a priori upper bound confirming that the energy 
scales like the area. The proof idea, relying deeply on the magnetic nature of the interaction, will be employed again several times. 

\begin{lemma}[\textbf{Trial upper bound}]\label{lem:trial}\mbox{}\\
	For any fixed bounded domain $\Omega$, and $\beta, \rho \ge 0$, 
	there exists a constant $C > 0$ s.t.
	$$
		\frac{E(L\Omega,\beta,\rho L^2|\Omega|)}{L^2} \le \frac{E_0(L\Omega,\beta,\rho L^2|\Omega|)}{L^2} \le C,
		\qquad \forall L \ge 1.
	$$
	%Equivalently, f
	%For any fixed simply connected domain $\Omega$ and $M \ge 0$, there exists a constant $C > 0$ s.t.
	%$$
		%\frac{E(\Omega,\beta,M)}{\beta} \le \frac{E_0(\Omega,\beta,M)}{\beta} \le C \qquad \forall \beta>1.
	%$$
\end{lemma}

\begin{proof}
	%The two statements are equivalent as per the scaling laws discussed in Lemma~\ref{lem:scale}. 
	That the Dirichlet energy is an upper bound to the Neumann energy is trivial 
	because $H^1_0(\Omega) \subseteq H^1(\Omega)$. Let us then prove the second inequality. 
	
	We fill the domain $L\Omega$ with $N \sim L^2$ subdomains on which
	we use fixed trial states with Dirichlet boundary conditions.
	The crucial observation is that the magnetic interactions between subdomains can be canceled by a suitable choice of phase (local gauge transformation). For concreteness we here take disks as our subdomains.
		
	Let $f \in C^\infty_c(B_1(0);\R^+)$ be a radial function with 
	$\int_{B_1(0)} |f|^2 = 1$, and let 
	$$
		u_j(\xv) := \sqrt{\omega_N} f(\xv-\xv_j) \in C^\infty_c(B_j),
		\qquad \omega_N := \rho L^2|\Omega|/N.
	$$ 
	Here the points $\xv_j$,
	$j=1,\ldots,N$, are distributed in $L\Omega$ in such a way that the disks 
	$B_j := B_1(\xv_j)$ are contained in $L\Omega$ and disjoint, with 
	$N \sim c|L\Omega|$ as $L \to \infty$ for some $c>0$. 
	Hence
	$$
	\lim_{N \to \infty} \omega_N = \rho/c.
	$$ 
	Take then the trial state
	$$
		u(\xv) := \sum_{j=1}^N u_j(\xv) e^{-i \beta\omega_N \sum_{k \neq j}\arg (\xv-\xv_k)}
		\ \in C^\infty_c(L\Omega).
	$$
	Note that its phase is smooth on each piece $B_j$ of its support and that
	$$
		|u(\xv)|^2 = \sum_{j=1}^N |u_j(\xv)|^2 =
		\begin{cases}
			|u_j(\xv)|^2,	& \text{for} \ \xv \in B_j, \\ 
			0,	&	\text{otherwise},
		\end{cases}
	$$
	and hence
	$$
		\int_{L\Omega} |u|^2 = N \omega_N = \rho L^2|\Omega|.
	$$
	Then
	\begin{align*}
		\cEAF_{\Omega,\beta} [u] 
		&= \sum_{j=1}^N \int_{B_j} \left| \left( -i\nabla + \beta\textstyle{\sum_{k=1}^N} \bA[|u_k|^2] \right) 
			e^{-i\beta\omega_N \sum_{k \neq j} \arg(\xv-\xv_k)} u_j(\xv) \right|^2 \diff \xv\\
		&= \sum_{j=1}^N \int_{B_j} \left| \Bigl( -i\nabla + \beta\bA[|u_j|^2] 
			+ \textstyle{\sum_{k \neq j}} 
				\bigl( \beta\bA[|u_k|^2] - \beta\omega_N \nabla \arg(\xv-\xv_k) \bigr) 
			\Bigr) u_j(\xv) \right|^2 \diff \xv \\
		&= \sum_{j=1}^N \int_{B_j} \left| (-i\nabla + \beta\bA[|u_j|^2]) u_j \right|^2
		= N \omega_N \int_{B_1(0)} |(-i\nabla + \beta\omega_N \bA[|f|^2])f|^2,
	\end{align*}
	where we used that by Newton's theorem~\cite[Theorem~9.7]{LieLos-01}
	$$
		\bA[|u_k|^2](\xv) = \nablap \int_{B_k} \ln|\xv-\yv| |u_k(\yv)|^2 \,\diff \yv
		= \nablap \ln|\xv-\xv_k| \int_{B_k} |u_k|^2 \,\diff \yv
		= \omega_N \nabla \arg(\xv-\xv_k),
	$$
	for $\xv \notin B_k$. It then follows that
	$$
		E_0(L\Omega,\beta,\rho L^2|\Omega|) \le \cEAF_{\Omega,\beta}[u]
		\le N\omega_N \bigl( \|\nabla f\|_{L^2} 
			+ \beta\omega_N \|\bA[|f|^2]f\|_{L^2} \bigr)^2
		\le C L^2,
	$$
	for some large enough constant $C>0$ independent of $N$ or $L$
	(but possibly depending on $\beta$, $\rho$ and $\Omega$).
\end{proof}

The following well-known inequalities will provide useful a priori bounds on the functional's minimizers:

\begin{lemma}[\textbf{Elementary magnetic inequalities}]\label{lem:mag ineq}\mbox{}\\
	\underline{Diamagnetic inequality}: 
	for any $\beta \in \R$ and $u\in H^1 (\Omega)$, 
	\begin{equation}\label{eq:diamag}
	\int_{\Omega}  \left| (\nabla + i\beta\bA[|u|^2])u \right|^2
			\ge \int_\Omega  \left|\nabla| u|\right|^2.
	\end{equation}
	\underline{Magnetic $L^4$ bound}: 
	for any $\beta \in \R$ and $u\in H ^1_0 (\Omega)$,
	\begin{equation}\label{eq:mag L4}
	 \int_{\Omega} \left| (\nabla + i\beta\bA[|u|^2])u \right|^2\ge 2\pi|\beta| \int_\Omega  |u|^4,
	\end{equation}
\end{lemma}

\begin{proof}
	The diamagnetic inequality is, e.g., given in~\cite[Theorem~7.21]{LieLos-01}, 
	while the $L^4$ bound follows immediately from the well-known inequality 
	\begin{equation}\label{eq:mag bound}
		\int_{\Omega} \left| (\nabla + i\bA )u \right|^2 
		\geq \pm \int_{\Omega} \curl \bA \, |u| ^2,
		\qquad u \in H^1_0(\Omega);
	\end{equation}
	see, e.g., \cite[Lemma~1.4.1]{FouHel-10}. 
	
	A proof of~\eqref{eq:mag bound} is to integrate the identity
	$$
		|(\nabla + i\bA)u|^2 
		= |( (\partial_1+iA_1) \pm i(\partial_2+iA_2) )u|^2 
			\pm \curl \,\bJ[u] \pm \bA \cdot \nablap |u|^2,
	$$
	with 
	$$ \bJ[u] : = \frac{i}{2} ( u \nabla \bar{u} - \bar{u} \nabla u). $$
	Thanks to the Dirichlet boundary conditions, the integral of the next to last 
	term vanishes while the last one can be integrated by parts yielding
	$$
		\mp \int_{\Omega} \curl \bA \, |u|^2.
	$$
	Again, no boundary terms are present because of the vanishing of $u$ on 
	$\partial \Omega$. Dirichlet boundary conditions are necessary since the bounds \eqref{eq:mag bound} resp. \eqref{eq:mag L4} 
	are otherwise invalid as $\bA \to 0$ resp. $\beta \to 0$, 
	as can be seen by taking the trial state $u \equiv 1$.
\end{proof}

In order to perform energy localizations we shall also need an IMS type inequality\footnote{The initials IMS may refer either to Israel Michael Sigal or to Ismagilov-Morgan-Simon.}, i.e. a suitable generalization of the well-known localization formula~\cite[Theorem~3.2]{CycFroKirSim-87}: 
	\begin{equation}\label{eq:IMS standard}
		|\nabla u|^2 = |\nabla(\chi u)|^2 + |\nabla(\eta u)|^2 
			- \left( |\nabla\chi|^2 + |\nabla\eta|^2 \right) |u|^2,
	\end{equation}
where $\chi ^2,\eta ^2$ form a partition of unity. 
	
\begin{lemma}[\textbf{IMS formula}]\label{lem:IMS}\mbox{}\\
	Let $\Omega \subseteq \R^2$ be a domain with Lipschitz boundary
	and $\chi^2 + \eta^2 = 1$ be a partition of unity s.t. 
	$\chi \in C_c^\infty(\Omega)$ 
	and $\supp\chi$ is simply connected.
%	and $\Omega \setminus \supp \eta$ is simply connected.
	Then, for any $u \in H^1(\Omega)$ and $\beta \in \R$,
	\bml{\label{eq:IMS}
		\cEAF_{\Omega,\beta}[u] 
		= \int_{\Omega} \left| (\nabla + i\beta\bA[|u|^2])(\chi u) \right|^2
		+ \int_{\Omega} \left| (\nabla + i\beta\bA[|u|^2])(\eta u) \right|^2	\\
		- \int_{\Omega} \left( |\nabla\chi|^2 + |\nabla\eta|^2 \right) |u|^2,
	}
	where
	\begin{equation}\label{eq:IMS-eta}
		\int_{\Omega} \left| (\nabla + i\beta\bA[|u|^2])(\eta u) \right|^2
		\ge \int_\Omega \left|\nabla|\eta u|\right|^2
	\end{equation}
	and
	\begin{equation}\label{eq:IMS-chi}
		\int_{\Omega} \left| (\nabla + i\beta\bA[|u|^2])(\chi u) \right|^2
		\ge \left\{ \begin{array}{l} \displaystyle
			\int_\Omega \left|\nabla|\chi u|\right|^2, \\ \displaystyle
			2\pi|\beta| \int_\Omega \chi^2 |u|^4, \\ \displaystyle
			(1-\eps)\cEAF_{\Omega,\beta}[\psi] 
				- (\eps^{-1}-1) \beta^2 \int_\Omega |\bA[|\eta u|^2\1_K]|^2 |\chi u|^2, 
%				- (\eps^{-1}-1) \beta^2 \int_\Omega |\tilde{\bA}|^2 |\chi u|^2, 
		\end{array} \right.
%		\cEAF_{\Omega,\beta}[u] 
%		&= \int_{\Omega} \left| (\nabla + i\beta\bA[|\psi|^2])(\psi u) \right|^2
%		+ \int_{\Omega} \left| (\nabla + i\beta\bA[|u|^2])(\eta u) \right|^2
%		- \int_{\Omega} \left( |\nabla\chi|^2 + |\nabla\eta|^2 \right) |u|^2
	\end{equation}
	with $\eps \in (0,1)$ arbitrary, $K := \supp\chi \cap \supp\eta$,
	and $\psi = e^{i\beta\phi}\chi u \in H^1_0(\supp \chi)$
	for some harmonic function $\phi \in C^2(\supp\chi)$.
%	and $\tilde{\bA} = \bA[|\eta u|^2] - \nabla \phi \in H^1_0(\supp \eta)$
%	for some harmonic function $\phi \in C^2(\bar\Omega)$.
\end{lemma}
\begin{proof}
	We expand
	$$
		\cEAF_{\Omega,\beta}[u] 
		= \int_\Omega |\nabla u|^2 
			+ 2\beta \int_\Omega \bA[|u|^2] \cdot \bJ[u]
			+ \beta^2 \int_\Omega \bigl|\bA[|u|^2]\bigr|^2 |u|^2.
	$$
	For the first term we use the standard IMS formula~\eqref{eq:IMS standard},
	while for the term involving $\bJ$ we have
	\begin{align*}
		\tx\frac{2}{i}(\bJ[\chi u] + \bJ[\eta u]) 
		&= u\chi\nabla(\chi \bar{u}) + u\eta\nabla(\eta \bar{u})
		- \bar{u}\chi\nabla(\chi u) - \bar{u}\eta\nabla(\eta u) \\
		&= u(\chi^2+\eta^2)\nabla \bar{u} - \bar{u} (\chi^2+\eta^2)\nabla u
		= \tx\frac{2}{i}\bJ[u].
%		\frac{2}{i}\bJ[u] 
%		=& u(\chi^2+\eta^2)\nabla\bar{u} - \bar{u}(\chi^2+\eta^2)\nabla u
%		= u\chi\nabla(\chi\bar{u}) - \chi(\nabla\chi)|u|^2
%		+ u\eta\nabla(\eta\bar{u}) - \eta(\nabla\eta)|u|^2 \\
%		&- \bar{u}\chi\nabla(\chi u) + \eta(\nabla\chi)|u|^2
%		- \bar{u}\eta\nabla(\eta u) + \eta(\nabla\eta)|u|^2
%		= \frac{2}{i}(\bJ[\chi u] + \bJ[\eta u]).
	\end{align*}
	We can then recollect the terms to obtain \eqref{eq:IMS}.
	Equation \eqref{eq:IMS-eta} and the first version of \eqref{eq:IMS-chi}  
	follow from the diamagnetic inequality~\eqref{eq:diamag}, while the second version of
	\eqref{eq:IMS-chi} follows from the magnetic bound~\eqref{eq:mag L4} with Dirichlet boundary conditions.
	For the third version we write
	\begin{multline*}
		\int_{\Omega} \left| (\nabla + i\beta\bA[|u|^2])(\chi u) \right|^2 \\
		= \int_{\Omega} \left| \left(\nabla + i\beta\bA[|\chi u|^2]
			+ i\beta\bA[|\eta u|^2 \1_K]
			+ i\beta\left(\bA[|\eta u|^2 \1_{K^c}] - \nabla\phi\right)\right)(e^{i\beta\phi} \chi u) \right|^2,
	\end{multline*}
	where the last magnetic term vanishes by taking the gauge choice
	$$
		\phi(\xv) := \int_{K^c} \arg (\xv-\yv) \,|\eta u(\yv)|^2 \,\diff \yv,
		\qquad \xv \in \supp\chi.
	$$
	Thus, noting that $|\chi u|^2 = |\psi|^2$, 
	$$ \int_{\Omega} \left| (\nabla + i\beta\bA[|u|^2])(\chi u) \right|^2 
		= \int_{\Omega} \left| \left(\nabla + i\beta\bA[|\psi|^2] \right)\psi
			+ i\beta\bA[|\eta u|^2 \1_K]\psi \right|^2,
	$$ 
	and there only remains to expand the square and bound the cross-term using Cauchy-Schwarz to conclude the proof.
\end{proof}

\subsection{Scaling laws}\label{sec:scale}

In fact the large $\beta$ and large volume limits are equivalent, 
as follows from the simple observation:

\begin{lemma}[\textbf{Scaling laws for the homogeneous gas}]\label{lem:scale}\mbox{}\\
	For any domain $\Omega \subset \R^2$ and $\lambda, \mu > 0$
	we have that
	\begin{equation}\label{eq:scale ener} 
	E(\Omega,\beta,M) = \frac{1}{\lambda^{2}} E\lf(\mu\Omega,\tx\frac{\beta}{\lambda^{2}\mu^{2}},\lambda^2\mu^2 M\ri),
	\end{equation}
	and an identical scaling relation %behavior 
	holds true for $E_0(\Omega, \beta, M)$.
\end{lemma}

\begin{proof}
	Given any $u \in H^1(\Omega)$ we may set
	\begin{equation}\label{eq:scale func}
	u_{\lambda,\mu}(\xv) := \lambda u(\xv/\mu),
	\end{equation}
	and observe that $u_{\lambda,\mu} \in H^1(\mu\Omega)$,
	$$\int_{\mu\Omega} |u_{\lambda,\mu}|^2 = \lambda^2 \mu^2 \int_{\Omega} |u|^2,$$
	and 
	$$\cEAF_{\mu\Omega,\beta}[u_{\lambda,\mu}] = \lambda^2 \cEAF_{\Omega,\beta \lambda^2 \mu^2}[u].$$
	Namely, using
	$\nablap w_0(\xv) = \xv^{-\perp} := \xv^\perp/|\xv|^2$
	and
	\begin{align*}
		\bA_{\mu\Omega}[|u_{\lambda,\mu}|^2](\xv) 
		&= \int_{\mu\Omega} (\xv-\yv)^{-\perp} |u_{\lambda,\mu}(\yv)|^2 \,\diff\yv
		= \lambda^2 \int_{\mu\Omega} (\xv-\yv)^{-\perp} |u(\yv/\mu)|^2 \,\diff\yv \\
		& = \lambda^2\mu \int_{\Omega} (\xv/\mu-\zv)^{-\perp} |u(\zv)|^2 \,\diff\zv
		= \lambda^2\mu \bA_{\Omega}[|u|^2](\xv/\mu),
	\end{align*}
	we have
	\begin{align*}
		\cEAF_{\mu\Omega,\beta}[u_{\lambda,\mu}] 
		&= \int_{\mu\Omega} \left| \nabla u_{\lambda,\mu}(\xv) + i\beta \bA_{\mu\Omega}[|u_{\lambda,\mu}|^2](\xv) u_{\lambda,\mu}(\xv) \right|^2 \diff\xv \\
		&= \int_{\mu\Omega} \left| \lambda\mu^{-1}(\nabla u)(\xv/\mu) + i\beta \lambda^3\mu \bA_{\Omega}[|u|^2](\xv/\mu) u(x/\mu) \right|^2 \diff\xv \\
		&= \lambda^2 \mu^{-2} \int_{\mu\Omega} \left| (\nabla u)(\xv/\mu) + i\beta \lambda^2\mu^2 \bA_{\Omega}[|u|^2](\xv/\mu) u(\xv/\mu) \right|^2 \diff\xv \\
		&= \lambda^2 \int_{\Omega} \left| \nabla u(\zv) + i\beta \lambda^2\mu^2 \bA_{\Omega}[|u|^2](\zv) u(\zv) \right|^2 \diff\zv
		\ = \lambda^2 \cEAF_{\Omega,\beta\lambda^2\mu^2}[u].
	%	= \int_{\Omega} \left| \left( \nabla + i\beta \bA[|u|^2] \right) u \right|^2
	\end{align*}
	Hence, we may take as a trial state for 
	$\cEAF_{\mu\Omega,\beta\lambda^{-2}\mu^{-2}}$
	the function $u_{\lambda,\mu}$ where $u$ is the minimizer 
	(or minimizing sequence)
	of $\cEAF_{\Omega,\beta}$, and vice versa.
	Moreover, if $u \in H^1_0$ then so is $u_{\lambda,\mu}$.
\end{proof}

It follows immediately from the above that the thermodynamic energy has a very 
simple dependence on its parameters, which justifies~\eqref{eq:scale thermo}
and the way it appears in~\eqref{eq:TF func}. 

\begin{corollary}[\textbf{Scaling laws for $ e(\beta,\rho) $}]\label{cor:scale}\mbox{}\\
	For any $\rho \geq 0$ and bounded $\Omega \subset \R^2$, with $e(\beta,\rho)$ defined as in \eqref{eq:thermo def}, we have
	\beq
		\label{eq:altern thermo}
		e(1, \rho) = |\Omega| \liminf_{\beta \to \infty} \frac{E_0 (\Omega, \beta, \rho)}{\beta},
	\eeq
	and	for any $\beta, \rho \ge 0$,
	\begin{equation}\label{eq:scale thermo bis}
		e(\beta,\rho) = \beta \rho ^2 e (1,1). 
	\end{equation}
\end{corollary}

\begin{remark}
	At the moment each shape of domain $\Omega$ may give rise to a different
	limit $e(\beta,\rho)$ in \eqref{eq:thermo def}, and this 
	Corollary and proof applies in such a situation.
	However it will be shown below in the case of Lipschitz regular domains that
	the limit is independent of the shape, and one may therefore w.l.o.g. take 
	the unit square $\Omega = Q$ as a reference domain.
\end{remark}

\begin{proof}
	A first consequence of the scaling property \eqref{eq:scale ener}
	is that taking the thermodynamic limit as described in \eqref{eq:thermo limit} 
	or \eqref{eq:thermo def} is equivalent to taking the limit $ \beta \to \infty $ 
	at fixed size of the domain, i.e.,
	\bdm
		e(c,\rho) = \liminf_{L\to \infty} \frac{E_0 (L \Omega, c, \rho |\Omega| L ^2)}{L ^2|\Omega|} 
		=  \liminf_{L\to \infty} \frac{E_0 (\Omega, c L^2 |\Omega|, \rho)}{L ^2},
	\edm
	where we have applied \eqref{eq:scale ener} with $\mu = L$, 
	$\lambda = |\Omega|^{1/2}$ and $M = \rho$. 
	Now if, for any $c > 0$, we set $\beta = c L^2 |\Omega| \to \infty$, 
	the above expression becomes
	\beq
		%\label{eq:altern thermo}
		e(c, \rho) = c |\Omega| \liminf_{\beta \to \infty} \frac{E_0 (\Omega, \beta, \rho)}{\beta},
	\eeq
	which proves the first claim, and also implies that
	\beq
		\label{eq:scale 1}
		e(c,\rho) = c \: e(1,\rho).
	\eeq
	
	Next we take $\mu = 1$ in~\eqref{eq:scale ener} and obtain
	$$E_0(\Omega,\beta,M) = \lambda^{-2} E_0(\Omega,\beta\lambda^{-2},\lambda^2 M).$$
	Taking $M = |\Omega|$, dividing by $|\Omega|$ and taking the limit $|\Omega| \to \infty$ we deduce 
	$$ e(\beta, 1) = \lambda^{-2} e (\beta\lambda^{-2}, \lambda ^2) = \lambda^{-4} e (\beta, \lambda ^2)$$
	where we used~\eqref{eq:scale 1} in the last equality. This yields
	\begin{equation}\label{eq:scale 2}
	 e(\beta, \rho) = \rho^2 \, e (\beta,1) 
	\end{equation}
	for all $\beta,\rho \geq 0$. Combining~\eqref{eq:scale 1} 
	and~\eqref{eq:scale 2} yields the result~\eqref{eq:scale thermo bis}.
\end{proof}

\subsection{Proof of Theorem~\ref{thm:thermo limit}}\label{sec:proof thermo}

We split the proof in three lemmas:

\begin{lemma}[\textbf{Thermodynamic limit for the Dirichlet energy in a square}]\label{lem:thermo Dir}\mbox{}\\
	Let $Q$ be a unit square, and $\rho \geq 0$ and $\beta \geq 0 $ be fixed constants. 
	The limit
	$$
		e(\beta,\rho) = \lim_{L\to +\infty} \frac{E_0 (L Q, \beta, \rho L ^2 )}{L ^2}
	$$
	exists, is finite, and satisfies 
	$e(\beta,\rho) \geq 2\pi \beta \rho^2$.
\end{lemma}

\begin{lemma}[\textbf{Neumann-Dirichlet comparison}]\label{lem:Neu Dir}\mbox{}\\
	Let $\Omega$ be a bounded simply connected domain with Lipschitz boundary, 
	then for any fixed $\rho$ and $\beta$ positive, as $L\to\infty$
	$$
		\frac{E_0 (L \Omega, \beta, \rho L ^2|\Omega|)}{L^2|\Omega|} \geq \frac{E (L \Omega, \beta, \rho L ^2|\Omega|)}{L ^2|\Omega|} \geq  \frac{E_0 (L \Omega, \beta, \rho L ^2|\Omega|)}{L ^2|\Omega|} - o (1).
	$$
\end{lemma}

\begin{lemma}[\textbf{Thermodynamic limit for the Dirichlet energy in a general domain}]
	\label{lem:bd domains}
	\mbox{}	\\
	Let $\Omega \subset \R^2$ be a bounded simply connected domain with 
	Lipschitz boundary, then
	\beq
		\label{eq:td limit bd domains}
		\lim_{L \to +\infty} \frac{E_0(L \Omega, \beta, \rho L^2|\Omega|)}{L^2 |\Omega|} 
		= e(\beta,\rho).
	\eeq
	%and, using the same notation as in Proposition \ref{pro:density} with $ -\frac{1}{2} < \eta < 0 $,
	%\beq
		%\label{eq:density bd domains}
		%\lim_{\beta \to \infty} \lf\|  \bar{\rho}^{\mathrm{af}}  - 1 \ri\|_{L^1(\Omega)} = 0.
	%\eeq
\end{lemma}

Theorem \ref{thm:thermo limit} immediately follows from these three results: 
combining Lemma~\ref{lem:thermo Dir} with Lemma~\ref{lem:Neu Dir} one obtains 
the existence of the thermodynamic limit for squares. In order to derive the 
result for general domains, one then uses Lemma~\ref{lem:bd domains} together 
with Lemma~\ref{lem:Neu Dir}. 
Notice that the proof of Lemma~\ref{lem:bd domains} 
requires only Lemma~\ref{lem:thermo Dir} and \ref{lem:Neu Dir} for squares
as key ingredients.

	\begin{figure}
		\begin{tikzpicture}
			%SQUARE
			\draw[thick] (0,0) rectangle (4.5, 4.5);
			%SUBSQUARES
			\draw (0,1) rectangle (3.5,4.5);
			\draw (0,2.75) -- (3.5,2.75);
			\draw (1.75,1) -- (1.75,4.5);
			%TEXTS
			\node [above left] at (0,0.1) {$L_m Q$};
			\node [above left] at (1.5,2.8) {$L_n Q_1$};
			\node [above left] at (3.0,2.8) {$\cdots$};
			\node [above left] at (1.5,1.05) {$L_n Q_j$};
			\node [above left] at (3.5,1) {$L_n Q_{q_{nm}^2}$};
			%FUNCTIONS
			\node [above left] at (5.3,0.6) {$u_0$};
			\draw (4.3,0.3) [xscale=2.1,<-] arc(-90:10:.30);
			\node [above left] at (-0.15,2.2) {$u_j$};
			\draw (0.7,1.9) [xscale=-4.0,<-] arc(-90:5:.30);
			%MEASURE
			\draw [<->] (3.5,2) -- (4.5,2);
			\node [above right] at (3.55,2) {$k_{nm}$};
		\end{tikzpicture}
		\caption{Filling the square $L_m Q$ with smaller squares $L_n Q_j$.}
		\label{fig:square-filling}
	\end{figure}
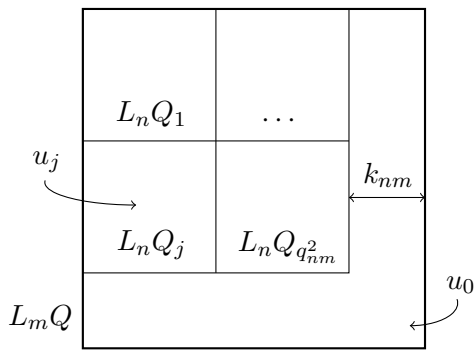

\begin{proof}[Proof of Lemma~\ref{lem:thermo Dir}]
From Lemma~\ref{lem:trial} we know that the sequence of energies per unit area 
has both an upper and lower limit. 
We denote $(L_n)_{n\in \NN}$ and $(L_m)_{m\in \NN}$ two increasing sequences of 
positive real numbers such that $L_n \to \infty$, $L_m\to \infty$ and
\begin{align*}
 \frac{E_0 (L_n Q, \beta, \rho L_n ^2 )}{L_n ^2} & \xrightarrow[n\to \infty]{} \liminf_{L\to \infty} \frac{E_0 (L Q, \beta, \rho L ^2 )}{L ^2},\\
 \frac{E_0 (L_m Q, \beta, \rho L_m ^2 )}{L_m ^2} &\xrightarrow[m\to \infty]{}  \limsup_{L\to \infty} \frac{E_0 (L Q, \beta, \rho L ^2 )}{L ^2}.
\end{align*}
For each $n$, there must exist a sequence of integers 
$$q_{nm}\xrightarrow[m\to \infty]{} +\infty$$
such that, for $m$ large enough, e.g., $ m \gg n $,
$$ 
	L_m = q_{nm} L_n + k_{nm}, \qquad 0\leq k_{nm} < L_n.
$$
We then build a trial state for $E_0 (L_{m} Q, \beta, \rho L_{m} ^2 )$ as 
follows (see Figure~\ref{fig:square-filling}). 
The square $L_{m} Q$ must contain $q_{nm} ^2$ disjoint squares of side-length 
$L_n$, that we denote $L_n Q_j, j=1, \ldots, q_{nm}^2$. Then we pick $u_j$ a 
minimizer of $E_0 (L_n Q_j, \beta, \rho L_n ^2 )$ and remark that by definition,
$$ 
	\sum_{k=1,k\neq j} ^{q_{nm}^2} \curl \bA [|u_k| ^2] = 0,
	\quad \mbox{ in } L_n Q_j.
$$
Thus there exists a gauge phase $\phi_j$ on the simply connected domain
$L_n Q_j$ such that 
$$ 
	\sum_{k=1,k\neq j} ^{q_{nm}^2} \bA [|u_k| ^2] = \nabla \phi_j,
	\quad  \mbox{ in } L_n Q_j.
$$
Similarly, there exists $\phi_0$ on the remaining part of the domain
(which can be arranged to be simply connected as well, as in Figure~\ref{fig:square-filling})
such that
$$ 
	\sum_{k=1} ^{q_{nm}^2} \bA [|u_k| ^2] = \nabla \phi_0,	
	\quad \mbox{ on } L_{m} Q \setminus \bigcup_{j=1} ^{q_{nm}^2} L_n Q_j.
$$
We define the trial state as (see the proof of Lemma~\ref{lem:trial})
$$ 
	u := \sum_{j=1} ^{q_{nm}^2} u_j e^{-i\beta\phi_j} + u_0 e^{-i\beta\phi_0}
$$
where $u_0$ is a function with compact support in 
$L_{m} Q \setminus \bigcup_{j=1} ^{q_{nm}^2} L_n Q_j$ satisfying 
$$ 
	\int_{L_{m}Q} |u_0| ^2 = \rho L^2_{m} - q_{nm}^2 \rho L_n ^2.
$$
By %Arguing as in the proof of 
Lemma~\ref{lem:trial}, we can construct $u_0$ 
such that 
$$
	\int_{L_m Q} \left|\left( \nabla + i\beta\bA [|u_0|^2] \right) u_0 \right| ^2 
	\leq C \bigl( L^2_{m} - q_{nm}^2 L_n ^2 \bigr) 
	\leq 2C L_m k_{nm}
$$
(where $C>0$ may depend on $\beta$ and $\rho$).
The function $u$ is an admissible trial state on $L_m Q$ because it is in $H^1$ 
on each subdomain, and continuous across boundaries 
due to the Dirichlet boundary conditions satisfied by each $u_j$. 
Computing the energy we have 
\begin{align*}
 \cEAF_{L_m Q,\beta} [u] &= \sum_{j=0} ^{q_{nm}^2} \int_{L_m Q} \left|e ^{-i\phi_j} \left( \nabla + i \beta\bA[ |u| ^2] - i \beta\nabla \phi_j \right) u_j \right| ^2 
 \\&=   \sum_{j=0} ^{q_{nm}^2} \int_{L_m Q} \left|\left( \nabla + i \beta\bA [ |u_j| ^2] \right) u_j \right| ^2 
 \\&= \sum_{j=1} ^{q_{nm}^2}  \cEAF_{L_n Q,\beta} [u_j] 
 	+ \int_{L_m Q} \left|\left( \nabla + i \beta\bA [ |u_0| ^2] \right) u_0 \right| ^2 
 \\&= q_{nm}^2 E_0 (L_n Q, \beta, \rho L_n ^2 ) + O(L_{m}k_{nm}),
 %+ O\left( L_{m}\left( L_{m} -  q_{nm} L_n\right)\right)
\end{align*}
with
$$
	q_{nm}^2 = \frac{L_m^2}{L_n^2}\left( 1 - \frac{k_{nm}}{L_m} \right)^2.
$$
Since $u$ has by definition mass $\rho L_{m} ^2$, it follows from the 
variational principle that
$$ 
	\frac{E_0 (L_{m} Q, \beta, \rho L_{m} ^2 )}{L_{m} ^2} 
	\leq \frac{E_0 (L_n Q, \beta, \rho L_n ^2 )}{L_n ^2} \left(1 + O\left(\frac{k_{nm}}{L_{m}} \right)\right) + O\left(\frac{k_{nm}}{L_{m}} \right).
%	\leq \frac{E_0 (L_n Q, \beta, \rho L_n ^2 )}{L_n ^2} \left(1 + O\left(\frac{ L_{m} -  q_{nm} L_n}{L_{m}} \right)\right) + O\left(\frac{ L_{m} -  q_{nm} L_n}{L_{m}} \right).
$$
Passing to the limit $m\to \infty$ first and then $n\to \infty$ yields
$$
	\limsup_{L\to \infty} \frac{E_0 (L Q, \beta, \rho L ^2 )}{L ^2} 
	\leq \liminf_{L\to \infty} \frac{E_0 (L Q, \beta, \rho L ^2 )}{L ^2},
$$
and thus the limit exists.

Additionally, we have by the bound \eqref{eq:mag L4},
$$
	\frac{1}{L^2} \cEAF_{LQ,\beta} [u] %\int_{LQ} \left| (\nabla + i\beta\bA[|u|^2])u \right|^2
	\ge \frac{2\pi\beta}{L^2} \int_{LQ} |u|^4
	\ge \frac{2\pi\beta}{L^4} \left( \int_{LQ} |u|^2 \right)^2
$$
for any $u \in H^1_0(LQ)$, proving that $e(\beta,\rho) \ge 2\pi\beta\rho^2$.
\end{proof}

\begin{proof}[Proof of Lemma~\ref{lem:Neu Dir}] 
	Since $H^1_0(\Omega) \subseteq H^1(\Omega)$ we obviously have 
	$$E_0(\Omega,\beta,M) \ge E(\Omega,\beta,M).$$
	Only the second inequality in the statement requires some work. 
	Let $u \in H^1(L\Omega)$ denote the minimizer of $\cEAF_{L\Omega,\beta}[u]$ 
	(see Proposition~\ref{prop:existence} of the Appendix) with mass 
	$$\int_{L\Omega} |u|^2 = \rho L^2 |\Omega|$$
	and no further constraint (thus satisfying Neumann boundary conditions). 
	In the sequel we take $\beta = 1$ and $|\Omega|=1$ to simplify the notation. 

	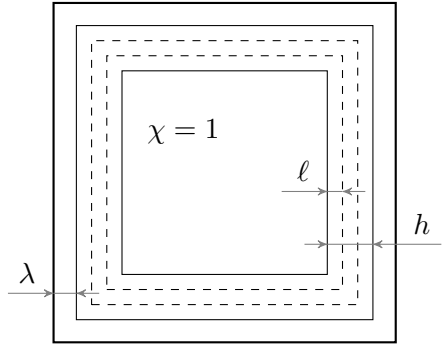
\begin{figure}
		\begin{tikzpicture}[>=stealth']
			%SQUARE
			\draw [thick] (0,0) rectangle (4.5, 4.5);
			%SHELLS
			\draw (0.3,0.3) rectangle (4.2,4.2);
			\draw [dashed] (0.5,0.5) rectangle (4,4);
			\draw [dashed] (0.7,0.7) rectangle (3.8,3.8);
			\draw (0.9,0.9) rectangle (3.6,3.6);
			%TEXTS & MEASURES
			\node [above left] at (-0.1,0.65) {$\lambda$};
			\draw [arrows=->,gray] 	(-0.6,0.65) -- (0.0,0.65);
			\draw [gray] 			(0.0,0.65) -- (0.3,0.65);
			\draw [arrows=<-,gray] 	(0.3,0.65) -- (0.6,0.65);
			\node [above right] at (4.6,1.3) {$h$};
			\draw [arrows=->,gray] 	(3.3,1.3) -- (3.6,1.3);
			\draw [gray] 			(3.6,1.3) -- (4.2,1.3);
			\draw [arrows=<-,gray] 	(4.2,1.3) -- (5.1,1.3);
			\node [above left] at (3.5,2.0) {$\ell$};
			\draw [arrows=->,gray] 	(3.0,2.0) -- (3.6,2.0);
			\draw [gray] 			(3.6,2.0) -- (3.8,2.0);
			\draw [arrows=<-,gray] 	(3.8,2.0) -- (4.1,2.0);
			\node [above right] at (1.1,2.5) {$\chi = 1$};
		\end{tikzpicture}
		\caption{Localizing on thin shells for the square $LQ$.}
		\label{fig:shells}
	\end{figure}
		
	We will need to make an IMS localization on a small enough region, and
	therefore consider a division of $L\Omega$ into a bulk region 
	surrounded by thin shells close to the boundary,
	where we will be using several different length scales 
	$L^{-1/3} \lesssim \lambda \ll 1 \ll L$
	and $L^{-1} \ll \ell \ll h \ll L$
	(see Figure~\ref{fig:shells} for the case of $\Omega = Q$ a square).
	
	We shall use Lemma~\ref{lem:IMS} a first time at distance $\lambda$ from 
	the boundary to deduce some useful a priori bounds. 
	Next, using a mean-value argument we show that, within a window of 
	thickness $h$ further from the boundary, there must exist one particular 
	shell of thickness $\ell$ where we have a good control on the mass and energy. 
	Finally we perform a second IMS localization with the truncation located in 
	this particular shell. This yields a lower bound in terms of the Dirichlet 
	energy in the bulk region, plus error terms that we can control using the a 
	priori bounds and in particular the good control on mass and energy in the 
	second localization shell.
	
	\medskip
	
	{\bf Step 1, a priori bounds.} 
	Let $\delta_\Omega(\xv) := \dist(\xv,\partial(L\Omega))$
	denote the distance function to the boundary, 
	which is Lipschitz and satisfies $|\nabla\delta_\Omega| \le 1$ a.e.
	We make a first partition of unity 
	$$\tilde{\chi}^2 + \tilde{\eta}^2=1$$
	such that $\tilde{\chi}$ varies
	smoothly from $1$ to $0$ on a shell $K_\lambda$ of width $\lambda$ 
	closest to the boundary of $L\Omega$, i.e.
	$K_\lambda := \{ \xv \in L\Omega : \delta_\Omega(\xv) < \lambda \}$.
	One may note that it is possible to construct these functions so as to satisfy
	$$
		|\nabla\tilde\chi| \le c \lambda^{-1} \tilde\chi ^{1-\mu}, 
		\quad |\nabla\tilde\eta| \le c\lambda ^{-1} \tilde\chi ^{1-\mu},
	$$
	for some arbitrarily small $\mu >0$, independent of $\lambda$, e.g., by 
	taking, in $\supp \tilde\chi \cap \supp \tilde\eta$, $\tilde\chi = f^{a}$ 
	and $\tilde\eta = \sqrt{1 - \tilde\chi^2}$ for $a$ large and some smooth 
	function $0 \leq f \leq 1$ varying on the right length scale and 
	reflection symmetric. 
	Then, by Lemma~\ref{lem:trial} and Lemma~\ref{lem:IMS},
	\begin{align}\label{eq:L4-norm-smallshell}
		CL^2 \ge \cEAF_{L\Omega,1}[u] 
		&\ge \int_{L\Omega} \Bigl( 2\pi\tilde\chi^2|u|^4 + |\nabla |\tilde\eta u| |^2 
			- \left( |\nabla\tilde\chi|^2 + |\nabla\tilde\eta|^2 \right)|u|^2 \Bigr) \nonumber \\
			&\ge \int_{L\Omega} \Bigl( 2\pi\tilde\chi^2|u|^4 + |\nabla |\tilde\eta u| |^2 
			- C \lambda ^{-2} \1_{K_\lambda} \tilde\chi ^{2-2\mu} |u|^2 \Bigr).
% 		\ge 2\pi\int_{L\Omega} \tilde\chi^2|u|^4 - c\lambda^{-2} \int_{K_\lambda}|\tilde\chi u|^2 \nonumber\\
% 		&\ge 2\pi\int_{L\Omega} |\tilde\chi u|^4 - c\lambda^{-2}(\lambda L)^{1/2} \left( \int_{K_\lambda}|\tilde\chi u|^4 \right)^{1/2}
% 		\ge \left( 2\pi - c\frac{ \lambda^{-3/2}L^{1/2} }{ (\int_{L\Omega}|\tilde\chi u|^4)^{1/2} } \right) 
% 			\int_{L\Omega} |\tilde\chi u|^4.
	\end{align}
	We bound the unwanted negative term as follows:
	\begin{align*}
	\lambda ^{-2} \int_{L\Omega} \1_{K_\lambda} \tilde\chi ^{2-2\mu} |u|^2 
	&\leq \lambda ^{-2} \left(\int_{K_\lambda} \tilde\chi ^{2-4\mu}\right)^{1/2} 
		\left( \int_{K_\lambda} \tilde\chi ^2 |u| ^4 \right) ^{1/2}
	\\&\leq C \lambda ^{-3/2} L ^{1/2} \left( \int_{K_\lambda} \tilde\chi ^2 |u| ^4   \right) ^{1/2} 
	\\&\leq C \delta L\lambda ^{-3} + C\delta ^{-1} \int_{L\Omega} \tilde\chi ^2 |u| ^4,
	\end{align*}
	with $\delta$ a fixed, large enough, constant. 
	Combining with~\eqref{eq:L4-norm-smallshell} we deduce 
	\begin{equation}\label{eq:Dir low a priori}
	 \int_{L\Omega} \left( 2\pi\tilde\chi^2|u|^4 + |\nabla |\tilde\eta u| |^2 \right) 
	 \leq C L^2 + C L \lambda^{-3} \leq C L ^2
	\end{equation}
	since we have chosen $\lambda \gtrsim L^{-1/3}$.
% , then either
% 	$$
% 		\int_{L\Omega} |\tilde\chi u|^4 \lesssim \lambda^{-3}L
% 		\lesssim L^2,
% 	$$
% 	or
% 	$$
% 		\int_{L\Omega} |\tilde\chi u|^4 \gg \lambda^{-3}L,
% 	$$
% 	which still implies
% 	$
% 		\int_{L\Omega} |\tilde\chi u|^4 \le C'L^2
% 	$
% 	by \eqref{eq:L4-norm-smallshell}.
	We note that this bound implies for the mass in a shell $K_\ell$
	of thickness $\ell$ in $L\Omega \setminus K_\lambda$ %$L\Omega \cap \{\tilde\chi=1\}$
	\begin{equation}\label{eq:smallshell mass rough}
		\int_{K_\ell} |u|^2 
		\le |K_\ell|^{1/2}\left( \int_{K_\ell} \tilde{\chi}^2 |u|^4 \right)^{1/2} 
		\lesssim \ell^{1/2} L^{3/2}.
	\end{equation}

	\medskip
	
	{\bf Step 2, finding a good shell.} We now select a region where the bounds~\eqref{eq:Dir low a priori} and~\eqref{eq:smallshell mass rough} can be improved. Consider dividing $L\Omega \setminus K_\lambda$ 
	into shells of thickness $\ell$ that form a layer closest to the shell $K_\lambda$,
	of total thickness $h \sim L^{1-\eps} \gg \ell$
	(again, see Figure~\ref{fig:shells}).
	Hence, we have 
	$$N_s := h/\ell \gg 1$$ such shells in the layer.
	Denote by $N_M$ the number of such shells $K_\ell$ with 
	$\int_{K_\ell} |u|^4 \ge M$.
	If $N_M < N_s$, there must exist a shell $K_\ell$ with
	$\int_{K_\ell} |u|^4 \le M$.
	But, using~\eqref{eq:Dir low a priori} and the fact that all the shells are included in the region where $\tilde\chi = 1$, we have
	$$
		M N_M \le \int_{L\Omega} \tilde\chi ^2 |u|^4 \le CL^2.
	$$
	We can thus ensure that $N_M < N_s$ by setting
	$$N_s = h/\ell \sim L^{1-\eps}\ell^{-1} \sim L^2/M,$$
	i.e. taking $M \sim \ell L^{1+\eps}$.
	Hence we have found a shell $K_\ell$ with
	\begin{equation}\label{eq:ener good shell}
		\int_{K_\ell} |u|^4 \leq C \ell L^{1+\eps},
	\end{equation}
	and thus
	\begin{equation}\label{eq:smallshell mass}
		\int_{K_\ell} |u|^2 \leq C (\ell L)^{1/2} (\ell L^{1+\eps})^{1/2} 
		= C \ell L^{1+\eps/2},
	\end{equation}
	improving \eqref{eq:smallshell mass rough}.

	\medskip
	
	{\bf Step 3, IMS localization in the good shell.} 
	We now perform a new magnetic localization on this $K_\ell$.
	We pick a partition $\chi^2+\eta^2=1$, 
	s.t. $\chi$ varies smoothly from $1$ to $0$ outwards on $K_\ell$, 
	%$ \supp \chi \subset K_{\ell} $ and $ \eta = 1 $ on $ K_{\ell}^c $. 
	so that $\chi=1$ resp. $\eta=1$ on the inner resp. outer component of $K_{\ell}^c$.
	Then, using Lemma~\ref{lem:IMS}, we have
	\bml{\label{eq:smallshell energy}
		\cEAF_{L\Omega,1}[u] 
		\ge (1-\delta)\cEAF_{L\Omega,1}[\psi] 
			- (\delta^{-1}-1) \int_{L\Omega} |\bA[|\eta u|^2\1_{K}]|^2 |\chi u|^2
			\\
			- \int_{K_\ell} (|\nabla\chi|^2 + |\nabla\eta|^2)|u|^2,
	}
	for any $\delta \in (0,1)$, where we have denoted 
	$\psi = \chi e^{i\phi}u$
	and $K = \supp\chi \cap \supp\eta \subseteq K_\ell$.
	Since $\psi$ is compactly supported in $L\Omega$ we have for the first term
	$$
		\cEAF_{L\Omega,1}[\psi] 
		\ge E_0 \lf(L\Omega,1,\lf\|\psi\ri\|_{L^2(L\Omega)}^2\ri) 
		= E_0 \lf(L\Omega,1,\lf\|\chi u\ri\|_{L^2(L\Omega)}^2\ri).
	$$
	Recalling the scaling relation~\eqref{eq:scale ener} 
	(taking $\mu = \lambda^{-1} = \tilde{L}/L$)
	and denoting 
	$$ 
		M = \int_{L\Omega} \chi ^2 |u| ^2,
		\qquad
		\tilde{L} = \sqrt{M/\rho},
	$$
	we have
	\begin{equation}\label{eq:scale proof thermo}
		E_0(L\Omega,1,M) %= \frac{M}{\rho L^2} E_0 (\sqrt{M/\rho}\Omega,1,M)
		= \frac{M}{\rho L^2} E_0 (\tilde{L}\Omega,1,\rho \tilde{L}^2). 
	\end{equation}
	We need to estimate the deviation of the mass $M$ of $\chi^2 |u|^2$
	from $\rho L^2 = \int_{L\Omega}|u|^2$:
	\begin{align}\label{eq:mass shell}
	 \left|\rho L ^2 -  \int_{L\Omega} \chi^2 |u| ^2 \right| &=\int_{L\Omega} \eta ^2 |u| ^2 = \int_{L\Omega} \tilde\eta ^2 |u| ^2 + \int_{L\Omega} \tilde\chi ^2 \eta^2 |u|^2\nonumber\\
	 &\leq C \lambda ^2 \int_{K_\lambda} |\nabla |\tilde \eta u| | ^2 + \left( \int_{L\Omega} \eta^2 \tilde\chi ^2 \right) ^{1/2} \left( \int_{L\Omega} \tilde\chi ^2 |u| ^4 \right) ^{1/2}\nonumber
	 \\&\leq C \lambda ^2 L ^2 + C h ^{1/2} L^{3/2} \ll L^2. 
	\end{align}
	Here we have used a Poincar\'e inequality to control the 
	$\tilde\eta ^2 |u| ^2$ term, making use of the fact that this function 
	vanishes at the inner boundary of $K_\lambda$. 
	It is not difficult 
	(see the proof methods of %for 
	\cite[Theorem~1 and 2 in Section~5.8.1]{Evans-98}
	%Theorem~1 and 2 in Section~5.8.1 of~\cite{Evans-98}
	and \cite[Theorem~8.11]{LieLos-01})
	%Theorem~8.11 in~\cite{LieLos-01}) 
	to realize that the constant involved in this 
	inequality applied on the set $K_\lambda$ can be taken to be proportional 
	to $\lambda ^2$. Note that $ \tilde L \to \infty $, if $ L \to \infty $, 
	thanks to \eqref{eq:mass shell}.
	Hence, inserting the above estimate in~\eqref{eq:scale proof thermo}, we get 
	\begin{equation}\label{eq:main term}
	 \frac{ \cEAF_{L\Omega,1}[\psi] }{L^2}
	 \geq \frac{ E_0(L\Omega,1,M) }{L^2} %= \frac{M^2}{(\rho L^2) ^2} \frac{E_0 (\sqrt{M}\Omega,1,\rho M)}{M}. 
	 = \frac{M^2}{(\rho L^2)^2} \frac{ E_0(\tilde{L}\Omega,1,\rho\tilde{L}^2) }{\tilde{L}^2}
	 = (1 + o(1)) \frac{ E_0(\tilde{L}\Omega,1,\rho\tilde{L}^2) }{\tilde{L}^2}.
	\end{equation}

	Then, there only remains to control the error terms in~\eqref{eq:smallshell energy}: 
	Using the H\"older and generalized Young inequalities  
	($\|\cdot\|_{p,w}$ denotes the weak-$L^p$ norm~\cite[Theorem~4.3,~Remarks]{LieLos-01}),
	\begin{multline*}
		\int_{L\Omega} \left| \bA[|\eta u|^2\1_{K}] \right|^2 |\chi u|^2
		\le \| \nabla w_0 \ast |\eta u|^2\1_{K} \|_{2p}^2 \|\chi u\|_{2q}^2
		\le c\|\nabla w_0\|_{2,w}^2 \|\eta u \1_{K}\|_{2r}^4 \|\chi u\|_{2q}^2 \\
		\le C\left( \int_{K_\ell} |\eta u|^{\frac{4q}{2q-1}} \right)^{\frac{2q-1}{q}} 
			\left( \int_{L\Omega} |\chi u|^{2q} \right)^{\frac{1}{q}},
	\end{multline*}
	where 
	$$\frac1p +\frac1q =1 \ \mbox{ and } \ 1+\frac{1}{2p} = \frac12 + \frac1r,$$
	i.e.,
	$$r = \frac{2q}{2q-1} \in (1,2) \ \mbox{ with } \ q \in (1,\infty).$$
	We can take $q=2$ and insert~\eqref{eq:Dir low a priori}-\eqref{eq:ener good shell} to obtain
	\begin{multline*}
		\left( \int_{K_\ell} |\eta u|^{8/3} \right)^{3/2} 
			\left( \int_{L\Omega} |\chi u|^4 \right)^{1/2}
		\le |K_\ell|^{1/2} \int_{K_\ell} |\eta u|^4 \left( \int_{L\Omega} |\chi u|^4 \right)^{1/2} \\
		\lesssim (\ell L)^{1/2} \ell L^{1+\eps} (L^2)^{1/2}
		= \ell^{3/2} L^{5/2+\eps}.
	\end{multline*}
	The last term in \eqref{eq:smallshell energy} is, 
	using \eqref{eq:smallshell mass}, bounded by
	$$
		c\ell^{-2} \int_{K_\ell} |u|^2 \lesssim \ell^{-1} L^{1+\eps/2}.
	$$
	There only remains to optimize the error terms in~\eqref{eq:smallshell energy}: 
	$$
		\delta E_0\lf(L\Omega,1,\lf\|\psi\ri\|_{L^2(L\Omega)}^2\ri)
			+ c_1(\delta^{-1}-1) \ell^{3/2} L^{5/2+\eps}
			+ c_2\ell^{-1} L^{1+\eps/2} 
		\le c_3 \delta L^2 + c_4 \delta^{-\frac{2}{5}} L^{\frac{8}{5}+\frac{7}{10}\eps},
	$$
	where we have picked $ \ell = L^{-3/5-\eps/5} \delta^{2/5} $, assuming that $ \delta \ll 1 $, as it will be.
	Thus, optimizing now over $ \delta $, i.e., taking $\delta \sim L^{-2/7+\eps/2}$, we have the bounds
	\beq
		\label{eq:error estimate}
		\frac{E_0(L\Omega,1,\rho L^2)}{L^2}
		\ge \frac{E(L\Omega,1,\rho L^2)}{L^2}
		\ge \frac{E_0\lf(L\Omega,1,\lf\|\psi\ri\|_{L^2(L\Omega)}^2\ri)}{L^2} - cL^{-\frac{2}{7}+\frac{\eps}{2}}.
	\eeq
	Combining with~\eqref{eq:main term} and passing to the liminf completes the proof.
\end{proof}

\begin{proof}[Proof of Lemma \ref{lem:bd domains}]
\label{sec:exist general domain} The result is proven as usual by comparing suitable upper and lower 
bounds to the energy.

\medskip

{\bf Step 1: upper bound.} We first cover $ L \Omega  $ with squares $ Q_j $, 
$ j = 1, \ldots, N_{\ell} $, of side length $ \ell = L^{\eta} $, 
$ 0 < \eta < 1 $, retaining only the squares $ Q_j $ completely contained in 
$ L \Omega $. 
One can estimate the area not covered by such squares as
\beq
	\lf| \Omega \setminus \lf( \bigcup_{j = 1}^{N_{\ell}} Q_j \ri) \ri| \leq C \ell L = o(L^2).
\eeq
Then we define the trial state
\beq
	u(\xv) : = \sum_{j = 1}^{N_{\ell}} u_j e^{- i\beta \phi_j},
\eeq
where 
\beq
	u_j(\xv) : = u_0(\xv - \xv_j) \1_{Q_j},
\eeq
with $ u_0 $ a minimizer of the Dirichlet problem 
with mass $ \rho L^2|\Omega|/N_{\ell} $
in a square $ Q $ with side length $ \ell $ centered at the origin,
and $\xv_j$ the center point of $Q_j$. 
The phases $ \phi_j $ are chosen in such a way that (see the proof of Lemma~\ref{lem:trial} again) 
$$ 
	\sum_{k=1,k\neq j}^{N_{\ell}} \bA [|u_k| ^2] = \nabla \phi_j,	\quad  \mbox{ in } Q_j.
$$
The existence of such phases is indeed guaranteed by the fact that
$$ 
	\sum_{k=1,k\neq j} ^{N_{\ell}} \curl \bA [|u_k| ^2] = 0,	\quad \mbox{ in }  Q_j.
$$
Hence
\bdm
	\cEAF_{L\Omega,\beta}[u] 
	= \sum_{j=1}^{N_{\ell}} \cEAF_{Q_j,\beta}[u_j] 
	= \sum_{j =1}^{N_{\ell}} E_0(\ell Q, \beta, \rho L^2|\Omega| N^{-1}_{\ell}),
\edm 
which implies
\begin{align} \label{eq:upper bound bd domains}
	\frac{E_0(L\Omega,\beta, \rho L^2)}{L^2|\Omega|} 
	&\leq \frac{1}{L^2|\Omega|} \sum_{j =1}^{N_{\ell}} E_0(\ell Q, \beta,\rho L^2|\Omega| N^{-1}_{\ell}) \\
	&= \frac{\ell^2}{L^2|\Omega|} \sum_{j =1}^{N_{\ell}} E_0(\ell Q, \beta, (1 + o(1)) \rho \ell^2)/\ell^2 
	= (1 + o(1)) e(\beta,\rho),
\end{align}
where we have estimated
\bdm
	N_{\ell} = \frac{|\bigcup_j Q_j|}{|Q_j|} 
	= \frac{(1 + o(1)) L^2|\Omega|}{\ell^2},
\edm
and used Lemma~\ref{lem:thermo Dir}. %Theorem \ref{thm:thermo limit}. 
Notice that, thanks to the assumption on 
$\eta$, we have $ \ell \to \infty $, which is crucial in order to apply 
Lemma~\ref{lem:thermo Dir}.

\medskip

{\bf Step 2: lower bound.} 
We again cover $L\Omega$ with squares $Q_{j=1,\ldots,N_\ell}$, 
this time keeping the full covering but still having 
$\ell^2 N_\ell / |L\Omega| \to 1$ as $L \to \infty$.
We pick a minimizer $\uAF = \uAF_L \in H^1_0(L\Omega)$ of $\cEAF_{L\Omega,\beta}$, with mass $\rho L^2|\Omega|$, and set
\beq
	\uAF_j : = \uAF \1_{Q_j},	
	\qquad \rho_j : = \dashint_{Q_j} \big| \uAF(\xv) \big|^2 \,\diff \xv.
\eeq
The idea of the proof is reminiscent of that in the upper bound part: we gauge away the magnetic interaction between the cells, and this leads to a lower bound in terms of the Neumann energy of the cells.

Note that $\uAF_j \in H^1(Q_j)$ for each $j$, and
$$
	\sum_{j = 1}^{N_{\ell}} \rho_j \ell^2 %= \frac{1 + o(1)}{\ell^2} = (1 + o(1)) \beta^{-2\eta}.
	= %(1 - o(1)) 
	\rho L^2|\Omega|.
$$
Before estimating the energy we need to distinguish between squares with 
sufficient mass and squares which will not contribute to the energy
to leading order. 
We thus set
\beq
	\mathcal{Q}_{L} : = \lf\{ Q_j,  j \in \{1, \ldots, N_{\ell} \} \: : \: \rho_j \geq L^{-2\eta+\delta} \ri\},
\eeq
for some $0 < \delta < 2\eta$. 
Note that the mass concentrated outside cells $\mathcal{Q}_{L}$ is relatively small:
\beq
	\sum_{Q_j \notin \mathcal{Q}_{L}} \rho_j \ell^2
	\leq C \ell^{2} N_\ell L^{-2\eta+\delta} = o(L^2).
\eeq

We can now estimate, using the gauge covariance of the functional on each $Q_j$,
\begin{align}
	E_0(L\Omega,\beta,\rho L^2|\Omega|)
	&= \cEAF_{L\Omega,\beta} [\uAF] 
	\geq \sum_{j=1}^{N_{\ell}} \int_{Q_j} \left| \left( -i\nabla + \beta \bA \left[|\uAF| ^2\right]\right) \uAF \right|^2 
	\nn
	\\&= \sum_{j=1}^{N_{\ell}} \int_{Q_j} \left| \left( -i\nabla + \beta \bA \left[|\uAF_j e^{i\beta\phi_j}| ^2\right]\right) \uAF_j e^{i\beta\phi_j} \right|^2  \nn
	\\&\geq \sum_{j=1}^{N_{\ell}} \rho_j\ell^2 \frac{E(\ell Q,\beta, \rho_j \ell^2)}{\rho_j\ell^2}  
	\geq \sum_{j : Q_j \in \mathcal{Q}_{L}} \rho_j^2\ell^2 \frac{E(\ell_j Q, \beta, \ell_j^2)}{\ell_j^2},
\end{align}
where $ \phi_j $ satisfies (observe that the left-hand side is curl-free on $Q_j$)
$$ 
	\sum_{k=1,k\neq j}^{N_{\ell}} \bA \left[|\uAF_k| ^2\right] = \nabla \phi_j,	\quad  \mbox{ in } Q_j,
$$
and in the last step we used the scaling law \eqref{eq:scale ener}
with $\mu = 1/\lambda = \sqrt{\rho_j}$.
Also,
$$ 
	\ell_j := \sqrt{\rho_j} \ell \ge L^{\delta/2}
	\xrightarrow[L \to \infty]{} + \infty
$$
uniformly in $j$ for cells $ Q_j \in \mathcal{Q}_{L} $,
and we thus conclude by Lemma~\ref{lem:thermo Dir} and \ref{lem:Neu Dir} that
\beq
	\label{eq:lower bound bd domains}
	\frac{1}{L^2|\Omega|} E_0(L\Omega,\beta,\rho L^2|\Omega|) 
	\geq (1-o(1)) \frac{e(\beta,1)}{L^2|\Omega|} \sum_{j : Q_j \in \cQ_{L}} \rho_j^2 \ell^2
	= (1-o(1)) \frac{e(\beta,1)}{L^2|\Omega|} \int_{\cQ} \bar\varrho^2,
\eeq
where we consider here the step function 
$\bar\varrho := \sum_{j : Q_j \in \cQ_L} \rho_j \1_{Q_j}$
and denote by $ \mathcal{Q} $ the union of the cells $\mathcal{Q}_{L}$.
It remains then to observe that
the constrained minimum
\bdm
	B = \min \lf\{ \int_{\mathcal{Q}} \varrho^2 \ : \ 0 \le \varrho \in L^2(\mathcal{Q}), \: \int_{\mathcal{Q}} \varrho = (1 - o(1))\rho L^2|\Omega| \ri\}.
\edm
is achieved by $\varrho$ constant and thus 
$$
	\int_{\cQ} \bar\varrho^2 \geq 
	B = \bigl( (1-o(1))\rho L^2|\Omega| \bigr)^2 |\mathcal{Q}|^{-1} 
	\geq (1-o(1))\rho^2 L^2|\Omega|.
$$ 
Inserting this in~\eqref{eq:lower bound bd domains} and using $\rho^2 e(\beta,1) = e(\beta,\rho)$ leads to the desired energy lower bound.
\end{proof}

\section{Proofs for the trapped gas}\label{sec:trap}

\subsection{Local density approximation: energy upper bound}\label{sec:up bound}

Here we prove the upper bound corresponding to~\eqref{eq:ener CV}:
\begin{equation}\label{eq:up bound}
	\EAF_\beta \leq \ETF_\beta \left( 1+o(1)\right),
	\qquad \text{as} \ \beta \to \infty.
\end{equation}
We start by covering the support of $ \rhoTF_{\beta} $ with squares $Q_j$, 
$j=1, \ldots, N_\beta $, centered at points $\xv_j$ and of side length $L$ with 
\begin{equation}\label{eq:up bound squares}
	L = \beta ^{\eta}, \qquad - \frac{s}{2(s + 2 )} < \eta < \frac{1}{s+2}.
\end{equation}
We choose the tiling in such a way that for any $ j = 1, \ldots, N_{\beta} $, $ Q_j \cap \supp(\rhoTF_{\beta}) \neq \emptyset $. 
The upper bound on $L$ indicates that the length scale of the tiling is much 
smaller than the size of the TF support. The lower bound ensures that it is 
much larger than the scale on which we expect the fine structure of the minimizer to live.

Our trial state is defined similarly as in the proof of Lemma~\ref{lem:bd domains}: 
\begin{equation}\label{eq:trap trial}
	\utest:= \sum_{j=1} ^{N_{\beta}} u_j e^{-i\beta\phi_j} 
\end{equation}
where $u_j$ realizes the Dirichlet infimum
$$ E_0 (Q_j, \beta, M_j) := \min \left\{ \cEAF_j [u] : u \in H^1_0 (Q_j), \: \int_{Q_j} |u| ^2 = M_j \right\},$$
where of course
$$
	\cEAF_j [u] = \cEAF_{Q_j,\beta}[u]
	= \int_{Q_j} \left| \left( -i\nabla + \beta \bA[|u|^2] \right) u \right| ^2
$$
and we set 
\begin{equation}\label{eq:def local mass}
	M_j = \int_{Q_j} |u_j| ^2:= \int_{Q_j} \rhoTF_\beta, %=: L^2 \rho_j.
	\qquad \rho_j := M_j/L^2 = \dashint_{Q_j} \rhoTF_\beta.
\end{equation}
The phase factors in~\eqref{eq:trap trial} are again defined so as to gauge away 
the interaction between cells, i.e., 
$$ \sum_{k=1,k\neq j} ^{N_\beta} \bA [|u_k| ^2] = \nabla \phi_j,	\quad \mbox{ in }  Q_j.$$
This construction yields an admissible trial state since $\utest$ is locally in 
$H^1$, continuous across cells by being zero on the boundaries, and clearly
$$
	\int_{\R^2} |\utest|^2 = \sum_{j =1}^{N_\beta} \int_{Q_j} |u_j|^2 
	= \sum_{j = 1}^{N_{\beta}} \int_{Q_j} \rhoTF_\beta = 1.
$$
Similarly as in the proofs of Lemmas~\ref{lem:trial} and~\ref{lem:bd domains} 
we thus obtain
\begin{equation}\label{eq:split up bound}
	\EAF_\beta \leq \cEAF_\beta [\utest] = \sum_{j = 1}^{N_{\beta}} \cEAF_j [u_j] + \int_{\R^2} V |\utest| ^2 
	= \sum_{j =1}^{N_{\beta}} E_0 (Q_j,\beta,M_j) + \int_{\R^2} V |\utest| ^2.
\end{equation}
Our task is then to estimate the right-hand side. 

We denote, for some $\eps >0$ small enough 
$$ S_\eps = \left\{\xv \in \supp(\rhoTF_{\beta}) \: \big| \: \rhoTF_\beta(\xv) \geq \beta ^{-\frac{2}{s+2} - \eps} \right\}$$
and split the above sum into two parts, distinguishing between cells fully 
included in $S_\eps$ and the others. Using~\eqref{eq:TF scaling}, 
it is clear that 
$$ \left| \supp\left(\rhoTF_\beta \right) \setminus S_\eps \right| \leq C \beta^{\frac{1}{s+2}} \cdot \beta ^{\frac{1}{s+2}-\eps}$$
where the first factor comes from the dilation transforming $\rhoTF_1$ into 
$\rhoTF_\beta$ and the second one is an estimate of the thickness of $S_\eps$ 
%based on~\eqref{eq:non degen}--\eqref{eq:TF estim}.  
based on~\eqref{eq:TF support}--\eqref{eq:non degen}.  

By a simple estimate of the potential $V$ in the vicinity of $S_\eps$ we %easily 
obtain
$$ 
	\sum_{j : Q_j\nsubseteq S_\eps} \int_{Q_j} V |u_j| ^2 
	\leq C \beta ^{\frac{s}{s+2}} \cdot \beta^{\frac{2}{s+2} - \eps} \cdot \beta ^{-\frac{2}{s+2} - \eps} 
	= C \beta ^{\frac{s}{s+2} - 2\eps} \ll \ETF_\beta,
$$
where the factor $\beta ^{\frac{s}{s+2}}$ accounts for the supremum of $V$, the factor $\beta^{\frac{2}{s+2} - \eps}$ for the volume of the 
integration domain and the factor $\beta ^{-\frac{2}{s+2} - \eps}$ for the typical value of $|u_j|^2$ on this 
domain. Also, using in addition Lemma ~\ref{lem:scale} and~\ref{lem:trial}, 
we deduce
$$ 
	\sum_{j : Q_j\nsubseteq S_\eps} E_0 (Q_j,\beta,M_j) 
	= \sum_{j : Q_j\nsubseteq S_\eps} E_0 (\beta^\eta Q,\beta,\beta^{2\eta} \rho_j) 
	\ll \ETF_\beta.
$$

For the main part of the sum in~\eqref{eq:split up bound} we use the scaling law 
(take $\lambda = \sqrt{\rho_j}$ and $\mu = \sqrt{\beta \rho_j}$ in 
Lemma~\ref{lem:scale}) to write  
$$ E_0 (Q_j,\beta,M_j) = \rho_j E_0 (L \sqrt{\beta \rho_j} Q, 1, L^2 \beta \rho_j) $$
with $Q$ the unit square. Then 
$$
	\sum_{j : Q_j \subseteq S_\eps} E_0 (Q_j,\beta,M_j) 
	= \sum_{j : Q_j \subseteq S_\eps} L^2 \beta \rho_j ^2 e(1,1) 
		+ \sum_{j : Q_j \subseteq S_\eps} L^2 \beta \rho_j ^2 \left(\frac{E_0(L_j Q,1,L_j^2)}{L_j^2} - e(1,1)\right)
$$
with, provided $\eps$ is suitably small 
and in view of the lower bound in~\eqref{eq:up bound squares} 
and the fact that we sum over squares included in $S_\eps$,
$$
	L_j := L \sqrt{\beta \rho_j} \ge \beta^{\eta + \frac{s}{2(s+2)} - \eps/2}
	\longrightarrow +\infty, 
	\quad \mbox{ uniformly with respect to } j = 1, \ldots N_\beta.
$$
We thus obtain (recall the definition of 
the thermodynamic energy in~\eqref{eq:thermo limit})
$$ \frac{E_0(L_j Q,1,L_j^2)}{L_j^2} \xrightarrow[L_j \to \infty]{} e(1,1)$$
uniformly in $j$, and deduce that
$$ 
	\sum_{j : Q_j \subset S_\eps} E_0 (Q_j,\beta,M_j) 
	= \left( 1+o(1) \right) \beta e(1,1) \sum_{j : Q_j \subset S_\eps} \rho_j^2\,L^2.
$$
Recalling that 
$$ \rho_j = \dashint_{Q_j} \rhoTF_\beta(\xv) \,\diff \xv,$$
we recognize a Riemann sum in the above. Using~\eqref{eq:TF estim} and the 
upper bound in~\eqref{eq:up bound squares} we may approximate 
$\rhoTF_\beta$ by a constant in each square
(this is most easily seen by rescaling to $\rhoTF_1$ and observing that
the size of squares then tends to zero), 
and bound the part of the integral 
located in the complement of $S_\eps$ similarly as above to conclude that 
$$ \sum_{j : Q_j \subset S_\eps} E_0 (Q_j,\beta,M_j) =  (1+o(1)) \beta e(1,1) \int_{\R^2} \lf( \rhoTF_\beta \ri)^2.$$
Using~\eqref{eq:asum V} and~\eqref{eq:TF support} we obtain
$$ |\nabla V(\xv)| \leq C \beta ^{\frac{s-1}{s+2}}$$
for any  $ \xv \in S_\eps$. Combining with~\eqref{eq:up bound squares} we deduce 
as above that  
$$ 
	\sum_{j : Q_j \subset S_\eps} \int_{Q_j} V|u_j| ^2 
	= (1+o(1)) \int_{\R ^2} V \rhoTF_\beta
$$
and this completes the proof of~\eqref{eq:up bound}.

\subsection{Local density approximation: energy lower bound}\label{sec:low bound}

Let us now complement~\eqref{eq:up bound} by proving the lower bound
\begin{equation}\label{eq:low bound}
\EAF_\beta \geq \ETF_\beta \left( 1+o(1)\right),
\end{equation}
thus completing the proof of~\eqref{eq:ener CV}. We again tile the plane with 
squares $Q_j, j=1, \ldots, N_\beta$, of side length 
$$L = \beta^{\eta} $$ 
satisfying \eqref{eq:up bound squares},
and taken to cover the finite disk $B_{\beta^t}(0)$ with
$$ t := \frac{1}{2+s} + \eps$$
for some $\eps>0$ to be chosen small enough. We also denote 
\begin{equation}\label{eq:good cells}
	\cQ_\beta :=\left\{Q_j \subset B_{\beta^t}(0) \big| \: L\sqrt{\rho_j \beta} \geq \beta^\mu \right\} 
\end{equation}
where $\uAF = \uAF_\beta$ is a minimizer for $\cEAF_\beta$ with unit mass and 
$$ \rho_j := \dashint_{Q_j} |\uAF(\xv)| ^2 \,\diff \xv.$$
Define the piecewise constant function
\begin{equation}\label{eq:def rho}
	\rhoBAF (\xv) := \sum_{Q_j \in \cQ_\beta} \rho_j \1_{Q_j}(\xv).
\end{equation}
We claim that one may find some $\mu >0$ in \eqref{eq:good cells}, such that %picking,
%\begin{equation}\label{eq:choice alpha}
%\alpha = \beta ^\mu 
%\end{equation}
%we have
\begin{equation}\label{eq:low mass}
	M:= \int_{\R ^2} \rhoBAF \xrightarrow[\beta \to \infty]{} 1. 
\end{equation}
Indeed, using~\eqref{eq:asum V} and~\eqref{eq:growth V} we get that for any 
$ \xv \in B^c_{\beta^t}(0) $
$$ V(\xv) \geq C \beta^{st} \min_{B^c_{1}(0)} V \geq C \beta^{st}$$
for $\beta$ large enough. Thus, using the energy upper bound~\eqref{eq:up bound} 
and dropping some positive terms we obtain
$$
	\beta^{st} \int_{B^c_{\beta^t}(0)} |\uAF| ^2 
	\leq \int_{\R ^2} V |\uAF| ^2 \leq \cEAF_\beta[\uAF] 
	\leq C \beta ^{\frac{s}{s+2}}
$$
and thus 
\begin{equation}\label{eq:low mass 1}
 \int_{B^c_{\beta^t}(0)} |\uAF| ^2 \leq C \beta ^{-s\eps}. 
\end{equation}
On the other hand, by definition of $\cQ_\beta$,
$$
	\sum_{Q_j\notin \cQ_\beta} \int_{Q_j} |\uAF| ^2 
	\leq N_{\beta} \beta^{2\mu-1}
$$
where $N_\beta$ is the total number of squares needed to tile $B_{\beta^t}(0)$. 
Clearly, we may estimate $N_\beta \leq C \beta ^{2t} L ^{-2} = C \beta^{2(t - \eta)}$
and then 
\begin{equation}\label{eq:low mass 2}
  \sum_{Q_j\notin \cQ_\beta} \int_{Q_j} |\uAF| ^2
  \leq C \beta^{2t - 2\eta + 2\mu - 1}  \ll 1
\end{equation}
because of~\eqref{eq:up bound squares}, which implies $-s/(s+2) - 2\eta < 0$, 
and provided we take $\eps$ and $\mu$ positive and small enough, e.g. (recall that $L= \beta ^\eta$ is the side-length of the tiling squares),
\begin{equation}\label{eq:choose mu}
	0 < \eps \leq \frac{1}{4}\lf( \frac{s}{s+2} + 2 \eta \ri),		
	\qquad	0 < \mu \leq \eps.
\end{equation}
Combining~\eqref{eq:low mass 1} and~\eqref{eq:low mass 2} and recalling that 
$\uAF$ is $L^2$-normalized proves~\eqref{eq:low mass}. 

With this in hand we turn to the energy lower bound per se. Let us again set 
$$\uAF_j = \uAF \1_{Q_j}, \qquad M_j = \rho_j L^2 = \int_{Q_j} |\uAF| ^2.$$
Dropping some positive terms we get
\begin{align}\label{eq:low first}
\EAF_\beta &= \cEAF_\beta [\uAF] \geq \sum_{Q_j\in \cQ_\beta} \int_{Q_j} \lf\{ \left| \left( -i\nabla + \beta \bA \left[|\uAF| ^2\right]\right) \uAF \right|^2 + V |\uAF| ^2 \ri\}
\nn
\\&= \sum_{Q_j\in \cQ_\beta} \int_{Q_j} \lf\{ \left| \left( -i\nabla + \beta \bA \left[|\uAF_j e^{i\beta\phi_j}| ^2\right]\right) \uAF_j e^{i\beta\phi_j} \right|^2 + V |\uAF_j| ^2 \ri\} \nn
\\&\geq \sum_{Q_j\in \cQ_\beta} \bigg\{ E(Q_j,\beta,M_j) + \int_{Q_j} V |\uAF_j| ^2 \bigg\} \nn
\\&\geq \sum_{Q_j\in \cQ_\beta} \bigg\{ \rho_j E\left(L\sqrt{\beta \rho_j} Q,1, (L\sqrt{\beta \rho_j}) ^2\right) +  \int_{Q_j} V |\uAF_j| ^2 \bigg\},
\end{align}
where the local gauge phase factors are defined as in previous arguments by 
demanding that (this is again possible because the left-hand side is 
$\curl$-free in the simply connected domain $Q_j$)
$$ \sum_{k=1,k\neq j} ^{N_\beta} \bA \left[|\uAF_k| ^2 \right] = \nabla \phi_j,  \qquad	\mbox{ in }  Q_j.$$
The minimum (Neumann) energy $E(Q_j,\beta,M_j)$ in the square $Q_j$ is defined 
as in~\eqref{eq:ener Neu} and we used the scaling laws following from 
Lemma~\ref{lem:scale} as previously to obtain 
$$E(Q_j,\beta,M_j) = \rho_j E\left(L\sqrt{\beta \rho_j} \, Q,1, (L\sqrt{\beta \rho_j}) ^2\right)$$  
with $Q$ the unit square. 
Next, we note that~\eqref{eq:up bound squares} and
\eqref{eq:good cells} %and~\eqref{eq:choice alpha} 
imply, using~\eqref{eq:choose mu},
%$\mu>0$, (which we are free to do),
$$ L_j = L\sqrt{\beta \rho_j} \ge \beta^\mu \longrightarrow \infty$$
uniformly in $j$ for all $j$ such that $Q_j\in\cQ_\beta$. 
Then, by Theorem~\ref{thm:thermo limit},
\begin{align*}
 \sum_{Q_j\in \cQ_\beta} & \rho_j E\left(L\sqrt{\beta \rho_j} Q,1, (L\sqrt{\beta \rho_j}) ^2\right) 
 = \sum_{Q_j\in \cQ_\beta} \beta L^2 \rho_j^2 \, E\left(L_j Q, 1, L_j^2\right)/L_j^2 \\
 & = (1+o(1)) \beta e(1,1) \sum_{Q_j\in \cQ_\beta} L^2 \rho_j^2 \nn 
 = (1+o(1)) \beta e(1,1) \int_{\R ^2} \big( \rhoBAF \big)^2. 
\end{align*}
On the other hand, it follows from~\eqref{eq:asum V} that, 
on all the squares of $\cQ_\beta$, 
$$|\nabla V| \leq C\beta^{\frac{s-1}{s+2} + \eps (s-1)}$$
and thus if
\beq
	\label{eq: tildeV}
	\tilde{V}(\xv) : = \sum_{Q_j \in \mathcal{Q}_{\beta}} V(\xv_j) \1_{Q_j}(\xv).
\eeq
we have that
$$
	\lf| V(\xv) - \tilde {V}(\xv)\ri| \leq C L \beta^{\frac{s-1}{s+2} + \eps (s-1)} = o(\ETF_\beta),
	\quad \mbox{ for any } \xv \in \mathcal{Q}_{\beta}.
$$
Recalling~\eqref{eq:def rho} and~\eqref{eq:low mass} we then have
\begin{align}
	\sum_{Q_j\in \cQ_\beta}  \int_{Q_j} V \,|\uAF_j| ^2 
	&= \int_{\R ^2} \tilde{V} \, \rhoBAF + O \left( L \beta^{\frac{s-1}{s+2} + \eps (s-1)}\right) 	\nn \\
	\label{eq:low pot}
	& = \int_{\R ^2} \tilde{V} \, \rhoBAF + o \left(\ETF_\beta\right).
\end{align}
%where we have set
%\beq
%	\label{eq: tildeV}
%	\tilde{V}(\xv) : = \sum_{Q_j \in \mathcal{Q}_{\alpha}} V(\xv_j) \1_{Q_j}(\xv).
%\eeq
The last assertion follows from~\eqref{eq:TF scaling} 
and~\eqref{eq:up bound squares}, provided we take $\eps$ small enough, 
e.g., for $ s > 1 $ (recall that the tiling squares have side-length $L= \beta ^\eta$),
\begin{equation}
	\eps \leq \frac{1}{2(s-1)} \lf(\frac{s-1}{s+2} + \eta \ri).
\end{equation}
In the very same way however we can put back $ V $ in place of $ \tilde{V} $, 
obtaining
\beq
	\label{eq:low int}
	\sum_{Q_j\in \cQ_\beta}  \int_{Q_j} V \,|\uAF_j| ^2 
	= \int_{\R ^2} \tilde{V} \, \rhoBAF + o \left(\ETF_\beta\right) 
	= \int_{\R^2} V \, \rhoBAF + o \left(\ETF_\beta\right),
\eeq

Combining~\eqref{eq:low first}, \eqref{eq:low pot} and \eqref{eq:low int} yields 
\begin{align}\label{eq:pre low}
	 \EAF_\beta & \geq \int_{\R ^2} V \: \rhoBAF 
	 	+ (1+o(1))  \beta e(1,1)\int_{\R ^2} \big( \rhoBAF \big) ^2 
	 	+ o (\ETF_\beta) \nn \\ 
	& \geq (1+o(1))\cETF_\beta[\rhoBAF] + o(\ETF_\beta) \nn \\
	& \geq (1+o(1))\ETF_\beta(M) + o(\ETF_\beta), 
%	\label{eq:tf lb}			
%	& \geq \ETF_{\beta (1+o(1))} \left( M \right) + o(\ETF_\beta),
\end{align}
where the latter energy denotes the ground state energy of the TF functional 
\eqref{eq:TF func}
minimized under the constraint that the $L^1$-norm be equal to $M$.
Inserting~\eqref{eq:low mass} and using explicit expressions as 
in~\eqref{eq:TF scaling} and~\eqref{eq:TF exp},
one obtains
$$ \ETF_{\beta} \left( M\right) = (1+o(1)) \ETF_\beta$$
in the limit $\beta \to \infty$, 
thus completing the proof of~\eqref{eq:low bound}.

\subsection{Density convergence}\label{sec:dens}

The lower bound in \eqref{eq:low bound} together with the energy upper bound~\eqref{eq:up bound}
implies that $\rhoBAF$, the piecewise constant approximation of $\rhoAF$ on scale $L=\beta^\eta$, 
is close in strong $ L^2 $ sense to $\rhoTF_\beta$. We will deduce~\eqref{eq:dens CV} from the following

\begin{lemma}[\textbf{Convergence of the piecewise approximation}]\label{lem:TF bound}\mbox{}\\
Let $\rhoBAF$ be defined as in~\eqref{eq:def rho} and $\rhoTF_\beta$ be the minimizer of~\eqref{eq:TF func}. Then  
\begin{equation}\label{eq:L2 estim}
\norm{\rhoBAF - \rhoTF_\beta}_{L^2(\R^2)} = o(\beta ^{-1/(s+2)}) 
\end{equation}
in the limit $\beta \to \infty$. 
\end{lemma}

\begin{proof}
Combining~\eqref{eq:up bound} and~\eqref{eq:pre low} we have 
\beq\label{eq:TF sandwich}
	\cETF_{\beta}[\rhoBAF] 
	\leq \EAF_\beta + o(1) \beta^{\frac{s}{s+2}} 
	\leq \ETF_{\beta} + o(1) \beta^{\frac{s}{s+2}}.
\eeq
% Recall also that the length scale $L = \beta ^{\eta}$ of the piecewise constant approximation $\rhoBAF$ is constrained by the requirement
% \beq
% 	\label{eq:eta cond}
% 	0 \leq \eta < \frac{s}{2(s + 2 )}.
% \eeq
The variational equation for $\rhoTF_\beta$ takes the form
$$ 
2 \beta e(1,1) \rhoTF_\beta + V = \lTF_\beta = \ETF_\beta + \beta e(1,1) \int_{\R ^2} (\rhoTF_\beta)^2 
$$
on the support of $\rhoTF_\beta$ (recall~\eqref{eq:TF exp} and~\eqref{eq:TF chem}). Thus, 
\begin{align*}
	\int_{\R^2} \lf( \rhoBAF - \rhoTF_{\beta} \ri)^2 
	&= \int_{\R^2} \left( (\rhoBAF)^2 + (\rhoTF_{\beta})^2 \right) - 2 \int_{\R ^2}\rhoBAF \rhoTF_{\beta} \\
	&= \int_{\R^2} (\rhoBAF)^2 + \int_{\R^2} (\rhoTF_{\beta})^2 - \frac{1}{\beta e(1,1)} \int_{\R^2} \rhoBAF \lf( \lTF_{\beta} - V \ri)_+  \\
	&\leq \frac{1}{\beta e(1,1)} \lf[ \cETF_\beta\left[\rhoBAF\right] - \lTF_{\beta} 
		+ \beta e(1,1) \int_{\R^2} (\rhoTF_{\beta})^2 \ri]\\
	&= \frac{1}{\beta e(1,1)} \lf[ \cETF_\beta\left[\rhoBAF\right] - \ETF_\beta \ri] = o(\beta^{-\frac{2}{s+2}}),
\end{align*}
where we used~\eqref{eq:TF sandwich} in the last step. 
\end{proof}

By the definition~\eqref{eq:def rho} of $\rhoBAF$ we also have, for any Lipschitz function $ \phi $ with compact support, 
\begin{align*}
	\int_{\R^2} \phi\left(\beta^{-1/(s+2)}\xv \right) \rhoBAF (\xv) \,\diff\xv
	& = \sum_{j = 1}^{N_{\beta}} \int_{Q_j} \phi\left(\beta^{-1/(s+2)} \xv\right) \: \rhoBAF \lf( \xv \ri) \diff\xv \\
	& = %\left(1 + O\left(\beta ^{\eta -\frac{1}{s+2}} \norm{\phi}_{\mathrm{Lip}} \right) \right) 
		\sum_{j = 1}^{N_{\beta}} \phi\left(\beta^{-1/(s+2)} \xv\right) \int_{Q_j} \rhoAF\lf( \xv \ri) \diff\xv 
		+ O\left(\beta ^{\eta -\frac{1}{s+2}} \norm{\phi}_{\mathrm{Lip}} \right) \\
	& = %\left(1 + O\left(\beta ^{\eta -\frac{1}{s+2}} \norm{\phi}_{\mathrm{Lip}} \right)\right) 
		\int_{\R^2} \phi\left(\beta^{-1/(s+2)} \xv\right) \: \rhoAF\lf( \xv \ri) \diff\xv 
		+ O\left(\beta ^{\eta -\frac{1}{s+2}} \norm{\phi}_{\mathrm{Lip}} \right),
\end{align*}
using the normalization of $\rhoAF$.
Furthermore, by Cauchy-Schwarz and Lemma~\ref{lem:TF bound} we obtain
%Inserting the result of Lemma~\ref{lem:TF bound} and using Cauchy-Schwarz we get
$$
	\int_{\R^2} \phi\left(\beta^{-1/(s+2)}\xv \right)
		\left( \rhoBAF(\xv) - \rhoTF_{\beta}(\xv) \right) \diff\xv
	= o(1) \norm{\phi}_{L^2(\R^2)}.
$$
%\begin{multline*}
% \left(1 + O\left(\beta ^{\eta -\frac{1}{s+2}} \norm{\phi}_{\mathrm{Lip}} \right)\right) \int_{\R^2} \phi\left(\beta^{-1/(s+2)}\xv \right) \rhoAF \lf( \xv \ri) \diff\xv \\ 
%	= \int_{\R^2} \phi\left(\beta^{-1/(s+2)}\xv \right)  \rhoTF_{\beta}(\xv) \,\diff\xv 
%	+ o(1) \norm{\phi}_{L^2}.
% 	+ o\left(\beta ^{- \frac{2}{s+2}}\right).
%\end{multline*}
Since the above estimates are uniform with respect to the Lipschitz norm of 
$\phi$, we can take $\eta < 1/(s+2)$, change scales in the above and 
recall~\eqref{eq:TF scaling} to deduce
$$
	\sup_{\substack{\phi \in C_0(B_R(0)) \\ \norm{\phi}_{\mathrm{Lip}} \le 1}}
	\left| \int_{\R^2} \phi(\xv) \left( \beta^{2/(s+2)} \rhoAF( \beta^{1/(s+2)}\xv ) - \rhoTF_1(\xv) \right) \diff\xv \right|
	= o(1), \quad \beta \to \infty,
$$
for fixed $R>0$,
hence~\eqref{eq:dens CV}.

\appendix

\section{Properties of minimizers}\label{sec:var eq}
%\section{Minimization of $ \cEAF $}\label{sec:var eq}

In this appendix we recall a few fundamental properties of the average-field 
functional \eqref{eq:avg func} in a trap $V$, 
respectively \eqref{eq:func dom} on a domain $\Omega$,
as well as their minimizers.

As discussed in \cite[Appendix]{LunRou-15}, the natural, maximal domain of 
$\cEAF$ is
$$
	\cDAF := \left\{ u \in H^1(\Omega) : \int_{\R^2} V|u|^2 < \infty \right\},
$$
and one may also use that the space $C^\infty_c(\R^2)$ is dense in this 
form domain w.r.t. %the graph of 
$\cEAF$.
Furthermore, \cite[Appendix: Proposition~3.7]{LunRou-15} ensures the existence
of a minimizer $\uAF \in \cDAF$ of $\cEAF_\beta$
for any value of $\beta \in \R$ for confining potentials $V$,
and by a similar proof and the compactness of the embedding 
$H^1_0(\Omega) \subset H^1(\Omega) \hookrightarrow L^p(\Omega)$, $1 \le p < \infty$,
the same holds for $\cEAF_\Omega$ for any bounded $\Omega$ with Lipschitz boundary:

\begin{proposition}[\textbf{Existence of minimizers}]\label{prop:existence}\mbox{}\\
	Let $\beta \in \R$ be arbitrary.
	Given any $V\colon \R^2 \to \R^+$ such that $-\Delta+V$ has compact resolvent, 
	there exists $\uAF \in \cDAF$ with 
	$\int_{\R^2}|\uAF|^2 = 1$ and $\cEAF_\beta[\uAF] = \EAF$.
	Moreover, if $M \ge 0$ and $\Omega \subset \R^2$ is
	bounded with Lipschitz boundary
%	such that $-\Delta_\Omega$ has compact resolvent
%	(in particular for $\Omega$ bounded with Lipschitz boundary) [Sobolev?]
	then there exists $\uAF \in H^1_{(0)}(\Omega)$ with 
	$\int_{\Omega}|\uAF|^2 = M$ and 
	$\cEAF_{\Omega,\beta}[\uAF] = E_{(0)}(\Omega,\beta,M)$.
\end{proposition}
\begin{proof}
	The first part is \cite[Appendix: Proposition~3.7]{LunRou-15}.
	For $\Omega \subset \R^2$ we have
	by the H\"older, weak Young, and Sobolev inequalities, 
	as well as Lemma~\ref{lem:mag ineq}, that
	\begin{align*}
		\|\bA[|u|^2]u\|_{L^2(\Omega)} 
		&\le \|\bA[|u|^2]\|_{L^4(\Omega)} \|u\|_{L^4(\Omega)}
		\le C\||u|^2\|_{L^{4/3}(\Omega)} \|\nabla w_0\|_{L^{2,w}(\R^2)} \|u\|_{L^4(\Omega)} \\
		&\le C'\||u|\|_{H^1(\Omega)}^3
		\le C'(M+\cEAF_\Omega[u])^{3/2},
	\end{align*}
	and therefore
	$$
		\|\nabla u\|_{L^2(\Omega)} 
		= \bigl\|\nabla u + i\beta\bA[|u|^2]u - i\beta\bA[|u|^2]u\bigr\|_{L^2(\Omega)}
		\le \cEAF[u]^{1/2} + C'|\beta|(M+\cEAF_\Omega[u])^{3/2}.
	$$
	Hence, given a minimizing sequence
	$$
		(u_n)_{n \to \infty} \subset H^1_{(0)}(\Omega),
		\quad \|u_n\|_{L^2(\Omega)}^2 = M,
		\quad \lim_{n \to \infty} \cEAF_\Omega[u_n] = E_{(0)}(\Omega,\beta,M),
	$$
	by uniform boundedness and the Rellich-Kondrachov theorem 
	(see, e.g., \cite[Theorem~8.9]{LieLos-01}) there exists
	a convergent subsequence (again denoted $u_n$) 
	and a limit $\uAF \in H^1_{(0)}(\Omega)$
	such that
	$$
		u_n \to \uAF \ \text{in} \ L^p(\Omega), \ 1 \le p < \infty,
		\qquad
		\nabla u_n \rightharpoonup \nabla \uAF \ \text{in} \ L^2(\Omega).
	$$
	Furthermore, by estimating
	$$
		\bigl\| \bA[|u_n|^2]u_n - \bA[|\uAF|^2]\uAF \bigr\|_2
		\le \bigl\| \bA[|u_n|^2 - |\uAF|^2]u_n \bigr\|_2 
			+ \bigl\| \bA[|\uAF|^2] (u_n-u) \bigr\|_2
	$$
	as above and using the strong convergence in 
	$L^p(\Omega)$ for any $1 \le p < \infty$, we have that
	$$
		\bA[|u_n|^2]u_n \to \bA[|\uAF|^2]\uAF \ \text{in} \ L^2(\Omega).
	$$
%	Furthermore, writing as in \cite{LunRou-15}
%	$$
%		\frac{1}{2}\cR(\xv,\yv,\zv)^{-2} 
%		= \sum_{\text{cyclic in $\xv,\yv,\zv$}} \frac{\xv-\yv}{|\xv-\yv|^2} \cdot \frac{\xv-\zv}{|\xv-\zv|^2} 
%		\quad \ge 0,
%	$$
%	we have that
%	\begin{align*}
%		\norm{\bA[|u_n|^2]u_n}_{L^2(\Omega)}^2
%		&= \frac{1}{6} \int_{\Omega^3} \cR(\xv,\yv,\zv)^{-2} |u_n(\xv)|^2 |u_n(\yv)|^2 |u_n(\zv)|^2 \,d\xv d\yv d\zv \\
%		&\to \frac{1}{6} \int_{\Omega^3} \cR(\xv,\yv,\zv)^{-2} |u(\xv)|^2 |u(\yv)|^2 |u(\zv)|^2 \,d\xv d\yv d\zv
%%		&= \frac{1}{6} \int_{\Omega^3} \cR(X)^{-2} \left| |u_n|^{\otimes 3} \right|^2 dX \\
%%		&\to \frac{1}{6} \int_{\Omega^3} \cR(X)^{-2} \left| |u|^{\otimes 3} \right|^2 dX 
%		= \norm{\bA[|u|^2]u}_{L^2(\Omega)}^2
%	\end{align*}
%	by dominated convergence. 
	Hence,
	\begin{align*}
		\norm{ (\nabla + i\beta \bA[|\uAF|^2])\uAF }_2 
		&= \sup_{\|v\|=1} |\langle \nabla \uAF + i\beta \bA[|\uAF|^2]\uAF, v \rangle| \\
		&= \sup_{\|v\|=1} \lim_{n \to \infty} |\langle \nabla u_n + i\beta \bA[|u_n|^2]u_n, v \rangle| \\
		&\le \liminf_{n \to \infty} \sup_{\|v\|=1} |\langle \nabla u_n + i\beta \bA[|u_n|^2]u_n, v \rangle| \\
		&= \liminf_{n \to \infty} \norm{ (\nabla + i\beta \bA[|u_n|^2])u_n }_2,
	\end{align*}
	that is 
	$E_{(0)}(\Omega,\beta,M) \le \cEAF_\Omega[\uAF] 
		\le \liminf_{n \to \infty} \cEAF_\Omega[u_n] = E_{(0)}(\Omega,\beta,M)$,
	and furthermore 
	$\int_\Omega |\uAF|^2 = \lim_{n \to \infty} \int_\Omega |u_n|^2 = M$.
\end{proof}

For completeness, we finish with a derivation of
the variational equation associated to the minimization of the energy functional \eqref{eq:avg func}.
%$$
	%\cEAF[u] = \int_{\R^2} \left( 
		%\left| \left( -i\nabla + \beta \bA[|u|^2] \right) u \right|^2
		%+ V|u|^2 \right).
%$$
Let us denote
$$
	\bJ[u] := \frac{i}{2}\left( u\nabla\bar{u} - \bar{u}\nabla u \right)
$$
and for two vector functions $\mathbf{F},\mathbf{G}: \R^2 \to \R^2$,
their convolution
$$
	\mathbf{F} \dotconv \mathbf{G}(\xv) := \int_{\R^2} \mathbf{F}(\xv-\yv) \cdot \mathbf{G}(\yv) \,\diff \yv.
$$

\begin{lemma}[\textbf{Variational equation}]\label{lem:var eq}\mbox{}\\
Let $u=\uAF$ be a solution to~\eqref{eq:GSE}. Then
\begin{equation} \label{eq:AF-equation}
	\left[ 
		\left( -i\nabla + \beta\bA[|u|^2] \right)^2 + V
		-2\beta \nablap w_0 \dotconv \left( \beta\bA[|u|^2]|u|^2 + \bJ[u] \right)
		\right]u = \lambda u,
\end{equation}
where
\begin{align}
	\lambda &= \cEAF[u] + \int_{\R^2} 
		\left( 2\beta \bA[|u|^2] \cdot \bJ[u] + 2\beta^2 |\bA[|u|^2]|^2 |u|^2 \right) \nonumber\\
	&= \int_{\R^2} \left( 
		1 (|\nabla u|^2 + V|u|^2)
		+ 2\cdot 2\beta \bA[|u|^2] \cdot \bJ[u]
		+ 3 \beta^2|\bA[|u|^2]|^2|u|^2
		\right).
	\label{eq:AF-lambda}
\end{align}
(Note that the factors 1, 2, resp. 3 correspond to the total degree of $|u|^2$ in each term.) 
\end{lemma}

\begin{proof}
Let
\begin{align*}
	\cF[u,\bar{u},\lambda] 
	&:= \cEAF[u,\bar{u}] + \lambda(1-\textstyle\int|u|^2) \\
	&= \int\left( 
		|\nabla u|^2 + (V-\lambda)|u|^2
		+ \beta^2|\bA[|u|^2]|^2|u|^2
		+ 2\beta \bA[|u|^2] \cdot \bJ[u]
		\right) + \lambda,
\end{align*}
$$
	\cE_1[u,\bar{u}] := \int |\bA[u\bar{u}]|^2 u\bar{u}
	= \iiint \nablap w_0(\xv-\yv) \cdot \nablap w_0(\xv-\zv) u\bar{u}(\xv) u\bar{u}(\yv) u\bar{u}(\zv) \,\diff\xv \diff\yv \diff\zv,
$$
$$
	\cE_2[u,\bar{u}] := \int \bA[u\bar{u}] \cdot i(u\nabla\bar{u} - \bar{u}\nabla u)
	= \iint \nablap w_0(\xv-\yv) u\bar{u}(\yv) \cdot i(u\nabla\bar{u} - \bar{u}\nabla u)(\xv) \,\diff\xv \diff\yv.
$$
We have
\begin{multline*}
	\cE_1[u,\bar{u}+\eps v] = \cE_1[u,\bar{u}] + \eps \iiint\Big(
		\nablap w_0(\xv-\yv) \cdot \nablap w_0(\xv-\zv) \big( 
		v(\xv)u(\xv) |u(\yv)|^2 |u(\zv)|^2 \\
		+ |u(\xv)|^2 u(\yv)v(\yv) |u(\zv)|^2 
		+ |u(\xv)|^2 |u(\yv)|^2 u(\zv)v(\zv) \big)
	\Big) \diff\xv \diff\yv \diff\zv + O(\eps^2),
\end{multline*}
hence at $O(\eps)$,
\begin{multline*}
	\int_\xv v(\xv)u(\xv)\bA[|u|^2]^2 \diff\xv 
	- \int_\yv v(\yv)u(\yv) \int_\xv \nablap w_0(\yv-\xv)|u(\xv)|^2 \cdot \int_\zv \nablap w_0(\xv-\zv)|u(\zv)|^2 \,\diff\zv \diff\xv \diff\yv \\
	- \int_\zv v(\zv)u(\zv) \int_\xv \nablap w_0(\zv-\xv)|u(\xv)|^2 \cdot \int_\yv \nablap w_0(\xv-\yv)|u(\yv)|^2 \,\diff\yv \diff\xv \diff\zv \\
	= \int vu\bA[|u|^2]^2 - 2\int vu \nablap w_0 \dotconv |u|^2\bA[|u|^2].
\end{multline*}
Also
\begin{multline*}
	\cE_2[u,\bar{u}+\eps v] = \cE_2[u,\bar{u}] + \eps \iint\Big(
		\nablap w_0(\xv-\yv) u(\yv)v(\yv) \cdot i(u\nabla\bar{u} - \bar{u}\nabla u)(\xv) \\
		+ \nablap w_0(\xv-\yv) |u(\yv)|^2 \cdot i\big( u(\xv)\nabla v(\xv) - v(\xv)\nabla u(\xv) \big)
	\Big) \diff\xv \diff\yv + O(\eps^2),
\end{multline*}
hence at $O(\eps)$ and using $\nabla \cdot \bA = 0$,
\begin{multline*}
	- \int_\yv v(\yv)u(\yv) \int_\xv \nablap w_0(\yv-\xv) \cdot 2\bJ[u](\xv) \,\diff\xv \diff\yv
	- i\int v(\xv)\nabla u(\xv) \cdot \bA[|u|^2](\xv) \,\diff\xv\\
	+ \underbrace{ i\int u(\xv)\bA[|u|^2](\xv) \cdot \nabla v(\xv) \,\diff\xv}_{
		= \{ PI \} = -i \int\nabla u \cdot \bA v - i\int u(\nabla\cdot\bA)v}
	= -2\int vu\nablap w_0 \dotconv \bJ[u] - 2i\int v\nabla u \cdot \bA[|u|^2]
\end{multline*}
Hence
\begin{multline*}
	\cF[u,\bar{u}+\eps v,\lambda] = \cF[u,\bar{u},\lambda] + \eps \int v\Big[
		(-\Delta + V-\lambda)u + \beta^2|\bA[|u|^2]|^2 u \\
		-2\beta^2 \nablap w_0 \dotconv |u|^2 \bA[|u|^2]u
		-2\beta \nablap w_0 \dotconv \bJ[u]u
		-2i\beta \bA[|u|^2] \cdot \nabla u
	\Big] + O(\eps^2),
\end{multline*}
and using
$$
	(-i\nabla + \beta\bA[|u|^2])^2 u 
	= -\Delta u - 2i\beta \bA[|u|^2] \cdot \nabla u + \beta^2 \bA[|u|^2]^2 u,
$$
we arrive at \eqref{eq:AF-equation}.

For \eqref{eq:AF-lambda} we use $\int |u|^2 =1$ by multiplying 
\eqref{eq:AF-equation} with $\bar{u}$ and integrating:
$$
	\lambda = \cEAF[u] - 2\beta\int |u|^2 \nablap w_0 \dotconv (\beta \bA[|u|^2]|u|^2 + \bJ[u]).
$$
We then use that
\begin{multline*}
	\int |u|^2 \nablap w_0 \dotconv \bA[|u|^2]|u|^2
	= \iiint |u(\xv)|^2 \nablap w_0(\xv-\yv) \cdot \nablap w_0(\yv-\zv) |u(\zv)|^2 |u(\yv)|^2 \,\diff\xv \diff\yv \diff\zv \\
	=-\iiint \nablap w_0(\yv-\xv) \cdot \nablap w_0(\yv-\zv) |u(\xv)|^2 |u(\zv)|^2 |u(\yv)|^2 \,\diff\xv \diff\zv \diff\yv
	= -\int \bA[|u|^2]^2 |u|^2
\end{multline*}
and 
\begin{multline*}
	2\int |u|^2 \nablap w_0 \dotconv \bJ[u]
	= \iint |u(\xv)|^2 \nablap w_0(\xv-\yv) \cdot i\left( u(\yv)\nabla \bar{u}(\yv) - \bar{u}(\yv)\nabla u(\yv) \right) \diff\xv \diff\yv\\
	= -\int_\yv i(u\nabla\bar{u} - \bar{u}\nabla u)(\yv) \cdot \int_\xv \nablap w_0(\yv-\xv) |u(\xv)|^2 \,\diff\xv \,\diff\yv
	= -2\int \bJ[u] \cdot \bA[|u|^2]
\end{multline*}
to arrive at \eqref{eq:AF-lambda}. 
\end{proof}

\providecommand{\bysame}{\leavevmode\hbox to3em{\hrulefill}\thinspace}
%\providecommand{\MR}{\relax\ifhmode\unskip\space\fi MR }
%\providecommand{\MRhref}[2]{%
%  \href{http://www.ams.org/mathscinet-getitem?mr=#1}{#2}
%}

\providecommand{\href}[2]{#2}

% 
% 
%  \bibliographystyle{amsalpha}
% %\bibliography{biblio_Oct16}
% \bibliography{/home/rougerie/Travail/Documentation/Bibtex/biblio_Oct16}
% 
% % \bibliography{biblio_DL}
% %\bibliography{biblio_NR_Fev16}
% 
% 

\end{document}